\documentclass[letterpaper,10pt,twocolumn]{IEEEtran}

\usepackage{amssymb,amsmath,amsfonts,bm,epsfig,graphicx,theorem}
\usepackage{rotating,setspace,latexsym,epsf,color,dsfont,caption}
\usepackage{cite,subfigure}
\usepackage{algorithm, algpseudocode}
\usepackage{epstopdf}
\usepackage{booktabs} 
\newtheorem{Theorem}{Theorem}
\newtheorem{Corollary}{Corollary}
\newtheorem{Lemma}{Lemma}

\newtheorem{Definition}{Definition}
\newtheorem{Assumptions}{Assumptions}

\newenvironment{proof}[1]{\medskip\par\noindent
{\bf Proof:\,}\,#1}{{\mbox{\,$\blacksquare$}\par}}

\newcommand{\Xv}{\mathbf{x}}

\newcommand{\Rb}{\mathbb{R}}
\newcommand{\Eb}{\mathbb{E}}
\newcommand{\Pb}{\mathbb{P}}

\newcommand{\lv}{\mathbf{1}}

\newcommand{\thetav}{\boldsymbol{\theta}}

\newcommand{\Tc}{\mathcal{T}}

\newcommand{\Sc}{\mathcal{S}}
\newcommand{\Cc}{\mathcal{C}}

\newcommand{\Fc}{\mathcal{F}}
\newcommand{\Pc}{\mathcal{P}}

\begin{document}
\IEEEoverridecommandlockouts

\title{Cost-Aware Learning and Optimization for Opportunistic Spectrum Access}

\author{Chao~Gan, 
	   Ruida~Zhou,
	    Jing~Yang,~\IEEEmembership{Member,~IEEE,}
        Cong~Shen,~\IEEEmembership{Senior~Member,~IEEE}
        % <-this % stops a space

\thanks{ C. Gan and J. Yang are with the
School of Electrical Engineering and Computer Science, The Pennsylvania State University, University Park, PA, USA (e-mail: cug203@psu.edu, yangjing@psu.edu).
}
\thanks{R. Zhou and C. Shen are with the Department of Electronic Engineering and Information Science, University of Science and Technology of China, Hefei, Anhui, China (e-mail: zrd127@mail.ustc.edu.cn, congshen@ustc.edu.cn).}
}

\maketitle
\thispagestyle{empty}

\begin{abstract}
	In this paper, we investigate cost-aware joint learning and optimization for multi-channel opportunistic spectrum access in a cognitive radio system. We investigate a discrete-time model where the time axis is partitioned into frames. Each frame consists of a sensing phase, followed by a transmission phase. During the sensing phase, the user is able to sense a subset of channels sequentially before it decides to use one of them in the following transmission phase. We assume the channel states alternate between {\it busy} and {\it idle} according to independent Bernoulli random processes from frame to frame. To capture the {\it inherent uncertainty} in channel sensing, we assume the reward of each transmission when the channel is idle is a random variable.  We also associate {\it random costs} with sensing and transmission actions. Our objective is to understand how the costs and reward of the actions would affect the optimal behavior of the user in both offline and online settings, and design the corresponding opportunistic spectrum access strategies to maximize the expected cumulative net reward (i.e., reward-minus-cost).
	
	We start with an offline setting where the statistics of the channel status, costs and reward are known beforehand.
	% Then, our objective is to design a joint learning and optimization algorithm to decide the actions of the transmitter within and across frames, so that the per-frame net reward (i.e., reward-minus-cost) is getting close to that under the optimal policy with a priori channel statistics as quickly as possible.
	We show that the the optimal policy exhibits a recursive double-threshold structure, and the user needs to compare the channel statistics with those thresholds sequentially in order to decide its actions. With such insights, we then study the online setting, where the statistical information of the channels, costs and reward are unknown a priori. We judiciously balance exploration and exploitation, and show that the cumulative regret scales in $O(\log T)$. We also establish a matched lower bound, which implies that our online algorithm is order-optimal. Simulation results corroborate our theoretical analysis.
	
	%, which depend on the costs and reward of actions, as well as the channel statistics
	%The optimal offline policy indicates that the user should examine the channels sequentially according to a descending order of $\theta_i$. For each channel $i$, there exist two thresholds: if $\theta_i$ is above the upper threshold, the user should transmit over $i$ directly without probing it; if it is below the lower threshold, the user should quit transmission; and if it lies in between, the user should sense the channel and access it if it is available, or move on to the next channel otherwise. We then investigate the online setting, and propose a joint learning and optimization algorithm to infer the channel statistics while tracking the optimal offline policy. We judiciously balance exploration and exploitation, and show that the accumulated regret scales in $O(\log T)$. We also establish a lower bound on the regret for any $\alpha$-consistent online strategy, which scales in $\Omega(\log T)$. This implies that our online algorithm is order-optimal. Simulation results corroborate our theoretical bounds.
	
\end{abstract}

\section{Introduction}
Advanced channel sensing technologies have enabled cognitive radio systems to acquire the channel status in real-time and exploit the temporal, spatial and spectral diversity of wireless communication channels for performance improvements~\cite{Liang:2011:CR}. Various opportunistic spectrum access strategies have been investiageted, under both offline settings where the channel statistics are known a priori~\cite{Kanodia:2004:MMO,Shu:2009:TSC,Sabharwal:2007:OSU,Liu:Infocom:2003,Andrews:2001:PQS,Borst:Infocom:2001,Chang:2009:OCP,Tan:2010:DOS}, and online settings where the users do not possess a priori channel statistics but will have to infer them from observations~\cite{Lai:CMA:2011,Lai:2008:MAC,Liu:2010:DMAB,Gai:dyspan:2010,Gai:2012:CNO,Liu:2015:OAD,Chen:2008:OSA,Zhao:2008:Myopic,Javidi:2008:Myopic,Liu:2008:Index,Ahmad:2009:MOA}.

While the main objective in such works is to improve the spectrum usage efficiency by leveraging the channel status measurements, the inherent uncertainty in channel sensing results, and the costs of channel sensing and transmission, are rarely investigated. In practice, even a channel is sensed to be idle, the transmission rate it can support is still uncertain, due to the inherent randomness of wireless medium. Thus, the {\it reward} of each transmission is random in general. Meanwhile, both sensing and transmission consume energy. Channel sensing also causes delay. For cognitive radio systems that operate under stringent energy and power constraints, or communication applications that can only tolerate short end-to-end delays, such {\it costs} become critical in determining the optimal operation of the cognitive radio system. Intuitively, the optimal spectrum access strategy depends on the intricate relationship between the costs and reward, as well as channel statistics. What makes the problem even more complicated is that the statistics of such quantities are often time-varying and unknown beforehand in the fast changing radio environment.

Within this context, in this paper, we investigate cost-aware learning and optimization for multi-channel opportunistic spectrum access in a cognitive radio system. Our objective is to analytically characterize the impact of the costs and reward of sensing and transmission on the optimal behavior of the user, and develop optimal cost-aware opportunistic spectrum access strategies in both offline and online settings. To this end, we adopt a discrete-time model, where the state of each channel evolves according to an independent Bernoulli process from time frame to time frame. The user is allowed to {\it sequentially} sense the channels at the beginning of each time frame to get measurements of the instantaneous channel states. The user then uses the measurements to decide its actions, i.e., whether to continue sensing, to transmit over one channel, or to quit the current frame. We associate random costs with sensing and transmission, and assign a positive random reward for each transmission. 
For the offline setting, we leverage the finite horizon dynamic programming (DP) formulation~\cite{Kumar:1986:SSE} to identify the optimal policy of the user to maximize the expected {\em net} reward (reward-minus-cost) in each frame. For the online setting, we cast the problem to the multi-armed bandit (MAB) framework~\cite{Lai:1985:AEA}, and propose a cost-aware online learning strategy to infer the statistics of the channels, costs and reward, and asymptotically achieve the maximum per-frame {\em net} reward.

\subsection{Main Contributions}
The main contributions of this paper are four-fold:

First, we identify the optimal offline spectrum access policy with a priori statistics of the channel states, the costs and reward of sensing and transmissions. The optimal offline policy exhibits a unique recursive double-threshold structure.
% i.e., the user will examine the channels sequentially according to a descending order of the channel statistics. Depending on the mean of a channel state, the user would take different actions: If the mean is above an upper threshold, the user should transmit over it without probing; If the mean is below another lower threshold, it should quit the current frame; And if the mean lies in between, the user should first probe the channel and then transmit over it if it is available, or continue to examine the next channel otherwise. 
The thresholds depend on the statistical information of the system, and can be determined in a recursive fashion. Such structural properties enable an efficient way to identify the optimal actions of the transceiver, and serve as the benchmark for the online algorithm developed in sequel.

Second, we propose an online algorithm to infer the statistics from past measurements, and track the optimal offline policy at the same time. In order to make the algorithm analytically tractable, we decouple the exploration stage and the exploitation stage. We judiciously control the length of the exploration stage to ensure that the sensing and transmission policy in the exploitation stage is identical to the optimal offline policy with high probability. We theoretically analyze the cumulative regret, and show that it scales in $O(\log T)$.

%During the exploration stage, the user probes the channels more aggressively, in order to obtain an accurate estimate of the channel statistics. During the exploitation stage, the user utilizes the estimated channel statistics to obtain the corresponding sensing and transmission policies.
Third, we establish a lower bound on the regret for any $\alpha$-consistent online strategies. $\alpha$-consistent strategies are those that perform reasonably well with high probability. The lower bound scales in $\Omega(\log T)$, which matches the upper bound in terms of scaling and implies that our online algorithm is order-optimal. To the best of our knowledge, this is the first online opportunistic spectrum access strategy achieving order-optimal regret when sequential sensing is considered.

Fourth, the online setting discussed in this paper is closely related to the cascading multi-armed bandits model~\cite{KvetonSWA15} in machine learning. The cost-aware learning strategy proposed in this paper extends the standard cost-oblivious cascading bandits in~\cite{KvetonSWA15}, and can be adapted and applied to a wide range of applications where the cost of pulling arms is non-negligible and the rank of arms affects the system performances, such as web search and dynamic medical treatment allocation. 

\subsection{Relation to the State of the Art}
Learning for multi-channel dynamic spectrum access is often cast in the MAB framework~\cite{Lai:1985:AEA}. In general, the classic non-Bayesian MAB assumes that there exist $K$ independent arms, each generating i.i.d. rewards over time from a given family of distributions with an unknown parameter. The objective is to play the arms for a time horizon $T$ to minimize the {\it regret}, i.e., the difference between the expected reward by always playing the best arm, and that without such prior knowledge. It has been shown that logarithmic regret is optimal \cite{Agrawal:1995,Auer:2002:FAM}. % Shipra {\it et al.} \cite{Shipra:TS:2012,Shipra:TS:2013} analyze the Bayesian principle based Thompson Sampling policy, and achieve logarithmnic regret as well. Empirically, Thompson Sampling outperforms other algorithms such as UCB \cite{Auer:2002:FAM}, $\epsilon$-greedy \cite{Kuleshov:2000}, or time-varying $\epsilon$-greedy~\cite{Auer:2002:FAM}.

Within the MAB framework, order-optimal sensing and transmission policies for both single-user scenario~\cite{Lai:CMA:2011} and multiple-user scenario in~\cite{Lai:2008:MAC,Liu:2010:DMAB,Gai:dyspan:2010,Gai:2012:CNO}. In those works, the objectives are mainly to identify the best channel or channel-user match and access them most of the time in order to maximize the expected throughput. 

Although the online strategy developed in this paper falls in the MAB framework, the sequential sensing model, and the intricate impact of the sensing and transmission costs and reward on the system operation make our problem significantly different from existing works \cite{Lai:CMA:2011,Lai:2008:MAC,Liu:2010:DMAB}. Since the user will stop sensing if certain condition is satisfied, the random stopping time implies that only a random subset of channels will be observed in each sensing phase. Such partial observation model makes the corresponding theoretical analysis very challenging. Moreover, the error in estimating the mean value of the costs and reward will affect the correctness of the online policy and propagate in a recursive fashion, which makes the regret analysis extremely difficult. 

% However, the sensing costs and the impact on the optimal operation of the transceiver are seldom investigated. 
%In \cite{Liu:2010:DMAB}, Liu and Zhao study decentralized medium access algorithms with multiple distributed players, and propose fair and order-optimal decentralized policies. Gai {\it et al.} consider geographically dispersed networks where different users observe different channel conditions, and propose a matching learning with polynomial storage algorithm in \cite{Gai:dyspan:2010,Gai:2012:CNO}. 

%Another different approach is to assume each channel evolves according to a Markov chain, and formulate the problem as a Partially Observable Markov Decision Process (POMDP) \cite{Kaelbling:1998:PAP,Smallwood:1973,Sondik:1978:OCP}, or restless bandits \cite{Chen:2008:OSA,Zhao:2008:Myopic,Javidi:2008:Myopic,Liu:2008:Index,Ahmad:2009:MOA,Liu:2008:Index}. Learning mechanisms for coordinating multiple users in this more complex setting are discussed in \cite{Liu:2008:Coop,Liu:2008:Dist,Liu:2011:Restless}.
The sequential sensing model and analytical approach adopted in this paper is similar to that in~\cite{Chang:2009:OCP,Liu:2015:OAD}.
In~\cite{Chang:2009:OCP}, a constant sensing cost is considered for each channel, and the optimal offline probing and transmission scheduling policy is obtained through DP formulation. The corresponding online algorithm is proposed in~\cite{Liu:2015:OAD}. Compared with \cite{Chang:2009:OCP,Liu:2015:OAD}, our model takes the randomness of the sensing/transmission costs and reward into consideration, which is a non-trivial extension. Besides, the Bernoulli channel status model adopted in our paper enables us to obtain the explicit structure of the optimal offline policy and the order-optimal online algorithm.

\subsection{Paper Outline}
This paper is organized as follows. Section \ref{sec:system_model} describes the system model. Section \ref{sec:offline} and Section~\ref{sec:online} describe the optimal offline policy and the online algorithm, respectively. Section \ref{sec:simulation} evaluates the proposed algorithms through simulations. Concluding remarks are provided in Section \ref{sec:discussion}. Important proofs are deferred to Appendix.

\section{System Model}\label{sec:system_model}
We consider a single wireless communication link consisting of $K$ channels, indexed by the set $[K] :=\{1,2,\ldots,K\}$ and a user who would like to send information to a receiver using exactly one of the channels. We partition the time axis into frames, where each frame consists of a channel sensing phase followed by a transmission phase. The channel sensing phase consists of multiple time slots, where in each slot, the user is able to sense one of the channels in $[K]$, and obtain a measurement of the instantaneous channel condition. Similar sequential sensing mechanism has been discussed in~\cite{Shu:2009:TSC,Sabharwal:2007:OSU,Chang:2009:OCP}. In this work, we assume that the sensing phase is {\em at most} $K$ time slots, which corresponds to the scenario that the user senses each channel once in the time frame. As we will see, the actual length of the probing phase depends on the parameters of the system, and will be automatically adjusted to optimize the system performance.  
We adopte the constant data time (CDT) model studied in \cite{Sabharwal:2007:OSU,Chang:2009:OCP}, where the transmitter has {\it a fixed amount} of time for data transmission, regardless of how many channels it senses. The length of the transmission phase is much larger than the duration of a time slot in the sensing phase in general. 

The communication over the link proceeds as follows. Within each frame, at the beginning of each time slot in the sensing phase, the user must choose between two actions: 1) {\it sense}: sample a channel that has not been sensed before and get its status, 2) {\it stop}: end the sensing phase in the current frame. Once the user stops sensing, it must decide between the following actions based on up-to-date sensing results:
1) {\it access}: transmitting over one of the channels already sensed in the current frame using a predefined transmission power. 2) {\it guess}: transmitting over one of the channels that have not been sensed in the current frame, 3) {\it quit}: giving up the current frame and wait until the next frame. 

%We assume there is a perfect feedback channel between the receiver and the transmitter, indicating whether the past transmission is successful. Thus, if the user chooses to transmit over channel $i$ in frame $t$ without probing (i.e. {\it guess}), the feedback information can be used to deduce the channel condition $X_i(t)$.  %All of the decisions are based on probing results in the current and previous frames.%, as well as previous transmission results.
\if{0}
\begin{figure}[t]
	\begin{center}
		\includegraphics[scale=0.65]{graph/frame.eps}
	\end{center}
	\vspace{-0.1in}
	\caption{Structure of a communication frame.}\label{fig:frame}
\end{figure}
\fi

In the following, we use $t=1,2,\ldots,T$ to index the time frames, and use $n=1,2,\ldots,N_t$ to index the time slots within a frame, where $N_t$ is the last time slot in frame $t$.  $(t,n)$ refers to the $n$-th time slot in the sensing phase of frame $t$. %We assume the distribution of $X_i(t)$ is independent across different channels. Thus, probing channel $i$ does not provide any information about the condition of any other channel in $[K]\backslash i$. %Similar assumptions have been made in \cite{Guha:2006:CISS,Sabharwal:2007:TON,Chang:2009:OCP,Liu:2015:OAD}.
Let $\alpha(t,n)\in[K]$ be the channel the user sensed at time $(t,n)$ for $n\leq N_t$. We use $\beta(t)\in[K]\cup\{0\}$ to denote channel the the user decides to user in the transmission phase in frame $t$. $\beta(t)=0$ indicates that it quits the transmission opportunity in frame $t$. 
Let $C_{\alpha(t,n)}$, $B_{\beta(t)}$, and $P_{\beta(t)}$ be the probing cost, communication reward, and transmission cost associated with the decisions $\alpha(t,n)$ and $\beta(t)$, respectively. Those costs can refer to the energy consumed for sensing/transmission, the interference caused by the actions, etc, and can be adjusted according to the resource constraints or quality of service (QoS) requirements in the system. The reward may correspond to the information bits successfully delivered during the transmission phase. We do not assign any cost for the action {\it quit} for a clear exposition of this paper. We can always extend our current model to include a positive cost for quit, which can be used to capture certain QoS requirements (such as delay) in the system. As we will see in the rest of this paper, this will not change the structure of the optimal offline policy, or the design and analysis of the online algorithm.

We make the following assumptions on the distributions of the channel statistics, the sensing and transmission costs, and the communication reward.
\begin{Assumptions}\label{assump:dist}
	\begin{enumerate}
		\item[1)] The state of each channel $i\in[K]$ stays constant within frame $t$ (denoted as $X_i(t)$), and varies across frames according to an independent and identically distributed (i.i.d.) Bernoulli random process with parameter $\theta_i$. Without loss of generality, we assume that $1\geq \theta_1\geq \theta_2\ldots\geq \theta_K> 0$. 
		\item[2)]  If $\beta(t)\in [K]$, and $X_{\beta(t)}(t)=1$, $B_{\beta(t)}$ is an i.i.d. random variable distributed over $[\underline{b},\underline{b}+\Delta_b]$ with mean $b_0$; Otherwise, if $\beta(t)=0$, or if $\beta(t)\in [K]$ and $X_{\beta(t)}(t)=0$, $B_{\beta(t)}=0$.
		\item[3)] $C_{\alpha(t,n)}$ is an i.i.d. random variable distributed over $[\underline{c},\underline{c}+\Delta_c]$ with mean $c_0$.
		\item[4)] If $\beta(t)\in[K]$, $P_{\beta(t)}$ is an i.i.d. random variable distributed over $[\underline{p},\underline{p}+\Delta_p]$ with mean $p_0$; If  $\beta(t)=0$, $P_{\beta(t)}=0$.
		\item[5)] $b_{0}- p_{0}>0$, and $\underline{b},\underline{c},\underline{p}, \Delta_b, \Delta_c, \Delta_p\geq 0$. 
	\end{enumerate}
	
\end{Assumptions}

Assumption~\ref{assump:dist}.1 indicates that the status of a channel alternates between two states: idle and busy, which is a common assumption in existing works~\cite{Shu:2009:TSC,Lai:2008:MAC,Lai:CMA:2011}. Assumption~\ref{assump:dist}.2 is related to the fact that the maximum transmission rate supported by an available channel is random, due to the uncertain link condition in wireless medium. Assumptions~\ref{assump:dist}.3 and \ref{assump:dist}.4 correspond to the sensing and transmission costs in the system. We assume they are random variables in general. In practice, the costs may be related to the physical resources (e.g., energy/power) available in the system, or QoS (e.g., delay) requirements of different applications. Therefore, they are usually not fixed but adaptively changing in order to satisfy the instantaneous constraints. When $\Delta_c=\Delta_p=0$, they become two positive constants. We impose Assumption~\ref{assump:dist}.5 to make the problem reasonable and non-trivial.

In the following, we will first identify the structure of the optimal offline policy with all of the statistical information in Assumptions~\ref{assump:dist} known a priori, and then develop an online scheme to learn the statistics and track the optimal offline policy progressively.

%Without knowledge of channel statistics, our objective is to design an online opportunistic spectrum access policy to automatically decide the operation of the transmitter, so that accumulative regret grows sublinearly in time.  While we expect that the online learning module is able to infer the channel statistics accurately with sufficient observations, searching for the optimal policy could still be computationally intensive, due to search space ``explosion". Therefore, in order to reduce the computational complexity, we will first identify the structure of the optimal policy with assumed channel statistics, and then rely on the online learning and decision-making modules to track the optimal policy progressively.

\section{Optimal Offline Policy}\label{sec:offline}
In this section, we assume that the statistics $\{\theta_i\}_{i=1}^K$, $b_0$, $c_0$, $p_0$ are known a priori. However, the instantaneous realizations of the corresponding random variables remain unknown until actions are taken and observations are made. Thus, the user needs to make sensing and transmission decisions based on up-to-date observations, as well as the statistics. In the following, we use {\it policy} to refer to the rules that the user would follow in a frame. Specifically, this includes an order to sample the channels, a stopping rule to stop probing (and determine $N_t$), and a transmission rule to decide which channel to use, all based on past measurements. We note that due to the randomness in the system, the same policy may lead to different observations. Accordingly, the user may take different actions in sequel.

Under Assumptions~\ref{assump:dist}, when all of the statistics are known beforehand, the observations made in one frame would not provide any extra information about other frames. Thus, the optimal offline policy should be the same in each frame. In this section, we will drop the frame index and focus on one individual frame. Let $\{\alpha(n),\beta\}_{n=1}^{N}$ be the sequence of sensing and transmission actions the user takes following the policy $\pi$. Then, the optimization problem can be formulated as follows:
\begin{align}\label{eqn:offline_reward}
& J^*=\max_{\pi} J^\pi = \max_{\pi}\Eb\left[\Eb\left[B_{\beta}- P_{\beta}\middle |\{X_{\alpha(n)}\}_{n=1}^{N}\right]\right]\nonumber\\
&\qquad\qquad\qquad-\sum_{n=1}^{N} \Eb\left[\Eb\left[C_{\alpha(n)}\middle | \{X_{\alpha(\tau)}\}_{\tau=1}^{n-1}\right] \right].
\end{align}
We will use $\pi^*$ to denote the optimal policy that achieves $J^*$. Such a policy is guaranteed to exist since there are a finite number of channels. We have the following observations.
\begin{Lemma}\label{lemma:probe}
	Under the optimal policy, if the transmitter senses a channel and finds it is available, it should stop sensing and then transmit over it in the upcoming transmission phase. % Otherwise, it should continue probing until it finds the first available channel or decides to {\it guess} or {\it quit}.
\end{Lemma}

%Due to the space limitation, we omit the proofs of Lemmas~\ref{lemma:probe}, \ref{lemma:ranking}, and Corollaries~\ref{cor:structure}-\ref{prop:maxN} in this paper. All of the omitted proofs can be found in the full version \cite{full}.

\begin{proof}
	First, we note that if the user transmits over a sampled and available channel, the first term in (\ref{eqn:offline_reward}) would be $b_{0}-p_{0}$, which cannot be improved further if the transmitter continues to {\it sense}, {\it guess} or {\it quit}. However, a continued sensing would increase the cost involved in the second term in (\ref{eqn:offline_reward}). Thus, transmitting over the channel would be the best action given that the transmitter samples a channel and finds it is available.
\end{proof}

\begin{Corollary}\label{cor:structure}
	There exist only two possible structures of the optimal policy: 1) The transmitter chooses to ``guess" without sensing any channel; 2) The user senses an ordered subset of channels sequentially, until it finds the first available channel. If none of them are available, it will decide to ``guess" or ``quit". Whenever the transmitter chooses to ``guess", it transmits over the best unsampled channel, i.e., the one with the maximum $\theta_i$.
\end{Corollary}

\begin{proof}
	Lemma~\ref{lemma:probe} indicates the structure of the optimal policy would be a pure ``guess", or a sequence of probing followed by a ``guess" or ``quit". It then becomes clear that the unsampled channel utilized for transmission should be the best unsampled channel, which gives the maximum expected reward.
\end{proof}

\begin{Corollary}\label{cor:state}
	Under the optimal policy, at any time slot, a sufficient information state is given by the tuple $(i,X_i,\Sc)$, where $\Sc$ is the set of unsensed channels and $i$ is the index of the last sensed channel in the frame. When $\Sc=[K]$, i.e., no channel has been sensed yet, the first two parameters equal zero.
\end{Corollary}

\begin{proof}
	When no channel has been examined yet, the information state can be represented as $(0,0,[K])$ without any ambiguity. When some channels have been sensed, state $(i, X_i,\Sc)$ implies that $X_j=0$, $\forall j\notin \Sc\cup\{i\}$. According to Corollary~\ref{cor:structure}, the probing order of those channels will not affect the reward/cost of any future actions, thus is redundant. Therefore, $(i,X_i,\Sc)$ contains sufficient information for future decision-making.
\end{proof}

Based on Corollary~\ref{cor:state}, in the following, we will use dynamic programming (DP) \cite{Kumar:1986:SSE} to represent the optimal decision process. Let $V (i, X_i,\Sc)$ denote the value function, i.e. maximum expected {\it remaining} net reward given the system state is $(i, X_{i},\Sc)$. Then,
\begin{align}\label{eqn:bellman}
& V(i, X_{i},\Sc)=\max\big\{  X_i b_0 \!- \! p_0, \,\max_{j\in\Sc} \{\!-\! c_0 \!+\!\Eb[V(j, X_j,\Sc\backslash j)]\},\nonumber\\
&\qquad\qquad\qquad\qquad \max_{j\in\Sc} \{\theta_jb_0- p_0\}, \,0 \big\} 
\end{align}
where the expectation is taken with respect to $X_j$.
The first term on the right hand side of (\ref{eqn:bellman}) represents the expected net reward if the user decides to {\it access} the last sensed channel, the second term represents the maximum expected net reward of probing a channel in $\Sc$, the third term represents the maximum expected net reward of {\it guess}, i.e., transmitting over the best unsampled channel, and the last term represents the net reward of {\it quitting} the current frame. Thus, given the information state $(i, X_i,\Sc)$, the user needs to choose the action to maximize the expected remaining net reward, while $V(0, 0,[K])$ gives the expected total net reward in a frame under the optimal offline policy. This is a standard finite-horizon DP, which can be solve through backward induction. Roughly speaking, without any structural information of the optimal solution, the state space of this DP is $O(K\times K!)$, which quickly becomes prohibitive as $K$ increases.
Therefore, in the following, we will first identify the structural properties of the optimal policy through theoretical analysis, and then leverage those properties to obtain the optimal offline policy.

We observe that the optimal policy has the following structure.
\begin{Lemma}\label{lemma:ranking}
	Denote $\Pc$ as the subset of channels to sense under the optimal policy. Then, the transmitter should sense the channels sequentially according to a descending order of $\theta_i$, until it finds the first available channel. Moreover, $\min_{i\in \Pc}\theta_i\geq  \max_{i\in [K]\backslash \Pc}\theta_i$, i.e., the channels in $\Pc$ are better than any other channel outside $\Pc$.
\end{Lemma}
The proof of Lemma~\ref{lemma:ranking} is given in Appendix~\ref{appx:lemma:ranking}.
% i.e., $X_k=0$ for all $\theta_k>\theta_j$
Lemma~\ref{lemma:ranking} indicates an efficient way to identify the optimal policy, as summarized in Theorem~\ref{thm:offline}.

\begin{Theorem}\label{thm:offline}%Assume $\theta_1\geq \theta_2\geq \ldots \geq \theta_K$.
	The optimal offline policy is to sequentially examine the channels starting from channel 1. At time slot $n$, there exist two thresholds $l_{n}$, $u_{n}$, $0\leq l_{n}\leq u_{n}\leq 1$, such that:
	\begin{itemize}%[noitemsep]
		\item if $\theta_{n}\in [u_{n},1]$, $\alpha(n)=0$, $\beta=n$, i.e., transmit over channel $n$ without probing.
		\item  if $\theta_{n}\in [l_{n}, u_{n})$, $\alpha(n)=n$, i.e., sense channel $n$. The user will transmit over channel $n$ if $X_{n}=1$, and it will move on to the next channel if $X_{n}=0$.
		\item if $\theta_{n}\in [0, l_{n})$, $\alpha(n)=\beta=0$, i.e., stop sensing and quit transmission in the current frame.
	\end{itemize}
	The values of the thresholds $u_n$ and $l_n$ can be determined recursively by solving (\ref{eqn:bellman}) through backward induction.
\end{Theorem}

The proof of Theorem~\ref{thm:offline} is presented in Appendix~\ref{appx:thm1}. 
Compared with the DP formulation in (\ref{eqn:bellman}), the state space of this offline policy now is reduced to $O(K)$. Thus, the optimal policy can be identified in a more computationally-efficient manner. 

{\bf Remark:} We point out that the optimal offline policy only depends on the mean of the costs and reward, thus, it can be directly applied to the scenario where the costs and reward are constants. Besides, it can also be applied to the case where the costs and reward are random but the instantaneous realizations are known beforehand. For this case, we can simply treat the costs and reward in each frame as constants and obtain the optimal policy in each frame individually. It can also be extended to handle the case where the costs and reward vary for different channels. 

\begin{Corollary}\label{prop:thresholds}
	$\{u_{i}\}_{i=1}^K$ is a monotonically decreasing sequence, and $\{l_{i}\}_{i=1}^K$ is a monotonically increasing sequence.
\end{Corollary}
Corollary~\ref{prop:thresholds} is implied by Lemma~\ref{lemma:E_i} in Appendix~\ref{appx:online} and the expressions of $l_i$ and $u_i$ in (\ref{eqn:a_i})(\ref{eqn:b_i}) in Appendix~\ref{appx:thm1}. 
The monotonicity of the thresholds indicates that $[l_K,u_K]\subseteq[l_{K-1},u_{K-1}]\ldots\subseteq [l_1,u_1]$, i.e., channels with higher $\theta_i$s are more prone to {\it sense}, while channels with lower $\theta_i$ are more prone to {\it guess} or {\it quit}. This is because the potential reward gain by sensing bad channel is small compared with the cost for sensing. 

The properties of the thresholds in Corollary~\ref{prop:thresholds} also coincide with the two optimal structures specified in Corollary~\ref{cor:structure}. 

\begin{Corollary}\label{prop:maxN}
	As $c_0$ increases, the maximum number of channels to sense in each frame decreases. 
\end{Corollary}
Corollary~\ref{prop:maxN} can be proved using the monotonicity of $E_i$ in $c_0$ according to (\ref{eqn:recursive}) in Appendix~\ref{appx:thm1}. It indicates that by adjusting the cost of probing, the user is able to adaptively choose the number of channels to sense in each frame, thus achieving the optimal tradeoff between the cost and reward incurred by the sensing actions. 

\section{Online Optimization with Learning}\label{sec:online}
In this section, we assume the statistics $\{\theta_i\}$, $c_0$, $b_0$, $p_0$ are unknown beforehand. Then, our
objective is to design an online strategy to decide $\{\alpha(t,n), \beta(t), N_t\}_{t,n}$ based on up-to-date observations, so as to minimize the following regret measure:
\begin{align}\label{eqn:regret}
& R(T)= TJ^*-\sum_{t=1}^T\Eb\left[\mathbb{E}\left[B_{\beta(t)}- P_{\beta(t)}\middle |\mathcal{F}^t\right]\right]\nonumber\\
&\qquad\quad-\sum_{t=1}^T\sum_{n=1}^{N_t}\Eb\left[ \mathbb{E}\hspace{-0.02in}\left[C_{\alpha(t,n)}\middle |\mathcal{F}^{t-1}, \{ X_{\alpha(t,\tau)}(t)\}_{\tau=1}^{n-1}\right]\right ],
\end{align}
where $\Fc_t:= \{X_{\alpha(t,n)}(t), C_{\alpha(t,n)}\}_{n=1}^{N_t}\cup B_{\beta(t)}\cup P_{\beta(t)} $, i.e., the states of the sensed channels, the costs of sensing those channel, and the corresponding transmission cost and reward; $\Fc^t:=\Fc_1\cup\ldots\cup \Fc_t$; and $J^*$ is the maximum expected reward as if the statistics were known beforehand. 

If $R(T)/T\rightarrow  0$, it is called sub-linear in total regret and zero-regret in time average. Our objective is to design an online strategy to make $R(T)/T$ converges to $0$ as quickly as possible.
Intuitively, as more measurements become available, the user is able to infer the channel statistics more accurately and make more informed decisions accordingly.  As we have observed in many previous works~\cite{Lai:CMA:2011,Lai:2008:MAC,Liu:2010:DMAB,Gai:dyspan:2010,Gai:2012:CNO}, the user faces an exploration-exploitation dilemma: On one hand, the user would take more sensing and transmission actions in order to get more measurements to refine its estimation accuracy; on the other hand, the user would exploit available information to track the optimal {\it offline} policy and optimize its net reward. The user should judiciously balance those two conflicting objectives in order to minimize the regret. 

What makes the problem different and much more challenging than those existing works is the recursive structure of the optimal offline policy. As a result, {\it the error in estimating $\{\theta_i\}, c_0, b_0, p_0$ would affect the ranking of the channels, as well as the decision that the user may take over each channel}. Tracking the impact of the estimation error on the overall regret thus becomes very complicated. Moreover, due to the randomness of the channel status realizations, the user would stop sensing after observing a random number of channels, even if the user sticks to the same offline policy. Thus, if the user tries to track the optimal offline policy during exploitation, {\it the channels are observed in a random and non-uniform fashion}: the channels ranked high are observed with larger probability, while the channels ranked low may have limited chance to be observed. Therefore, the algorithm should take the sampling bias into consideration and adjust the exploration in a more sophisticated fashion.

We detail our joint learning and optimization strategy in Algorithm~\ref{alg:online}. In order to tackle the aforementioned challenges, we decouple exploration and exploitation to two separate phases. We use $T_i(t)$ to track the number of times that channel $i$ has been sampled during exploration phase up to frame $t$. In the exploration phase, the channels that have been sampled less than $D(t) = L \log t+ D$ times up to frame $t$ will be sampled, with $L, D$ being positive constant parameters. The specific values of $L$ and $D$ to ensure the optimal convergence rate of $R(T)/T$ will be discussed in Lemma~\ref{summary} in Appendix~\ref{appx:online}. In the exploitation phase, the user first estimate $\{\theta_i\}, c_0, b_0, p_0$ (denoted as $\{\hat{\theta}_i\}, \hat{c}_0, \hat{b}_0, \hat{p}_0$) by calculating the empirical average of the channel states, the sensing cost, the transmission reward and cost based on historical observations $\Fc^{t-1}$. It then executes the optimal offline policy using the estimates according to Theorem~\ref{thm:offline}. 

\begin{algorithm}[t]
	\caption{Joint Learning and Optimization.}
	\begin{algorithmic}[1]
		\State {\bf Initialization}: Sense each channel at $t=1$; Set $T_i (1) = 1$ for all $i\in[K]$; Record $\Fc_1$.%=\cup_{j}\{X_j(k)\}_{k=1}^{T_j(t)}$.
		\While {$t$}
		\State $t:=t+1$;
		\State Let $\mathcal{E}(t):=\{i:T_i(t)<D(t)\}$.
		\If {$\mathcal{E}(t)\neq \emptyset$} \Comment{{\it Exploration}}
		\State Sample every channel in $\mathcal{E}(t)$.
		\State Transmit over one available channel in $\mathcal{E}(t)$ if there is any.
		\Else  \Comment{{\it Exploitation}}
		\State Calculate empirical averages $\{\hat{\theta}_i\}$, $\hat{c}_0$, $\hat{b}_0$, $\hat{p}_0$. 
		\State  Sort channels with $\hat{\theta}_i$. %:=\frac{\sum_{k=1}^{T_i(t)}X_i(k)}{T_i(t)}$
		\State Calculate $\{\hat{u}_i\}$, $\{\hat{l}_i\}$ in Theorem~\ref{thm:offline} using $\{\hat{\theta}_i\}$, $\hat{c}_0$, $\hat{b}_0$, $\hat{p}_0$.
		\State Take actions according to Theorem~\ref{thm:offline}.
		\EndIf
		\State Update $T_i(t)$, $\Fc_t$.
		\EndWhile
	\end{algorithmic}\label{alg:online}
\end{algorithm}

The performance of the online algorithm is theoretically characterized in the following theorem.
\begin{Theorem}\label{thm2}
	There exist sufficiently large constants $L$, $D$ such that the regret under Algorithm~\ref{alg:online} is bounded by
	\begin{align}
	R(T) &\leq  C_{1} \log T +C_2,
	\end{align}
	where $C_1=KL(b_{0}+K c_{0})$, $C_2=\left(\pi^2+DK+1\right)(b_{0}+K c_{0})$.
\end{Theorem}

Theorem~\ref{thm2} indicates that the cumulative regret $R(T)$ is sub-linear in $T$, thus achieving zero-regret asymptotically. The first term of the regret is due to the exploration while the last term comes from exploitation. The proof is sketched as follows: we first relate the error in estimating the thresholds $\{u_{i}\}, \{l_{i}\}$ with the estimation errors $\{\theta_{i}-\hat{\theta}_i\}$, $c_0-\hat{c}_0$, $b_0-\hat{b}_0$, $p_0-\hat{p}_0$. Based on this relationship, we derive an upper bound on the number of samples required to ensure all estimation errors are sufficiently small, so that the user will choose the optimal offline policy in the exploitation phase with high probability. Finally, we explicitly bound the regret by examining the regrets from exploration and exploitation separately. A detailed proof of Theorem~\ref{thm2} is presented in Appendix~\ref{appx:online}. 

We also establish a lower bound on the regret under any $\alpha$-consistent~\cite{Lai:1985:AEA} online strategy in the following. 
\begin{Definition}[$\alpha$-consistent strategy]\label{dfn:alpha_consistent} Let $N^*(T)$ be the number of times that the user takes the optimal offline policy in a frame over the first $T$ frames under an online strategy. Then, $\forall \alpha \in (0,1)$, if $$\frac{\Eb[T - N^*(T)]}{T^{\alpha}} = o(1),$$ the strategy is $\alpha$-consistent.
\end{Definition}
Roughly speaking, an $\alpha$-consistent strategy represents a category of ``good" online strategies under which the user selects the optimal offline policy in each frame with high probability. 

\begin{figure*}
	\centering
	\subfigure[Cumulative regret ]{ \includegraphics[width=0.4\textwidth, scale=1]{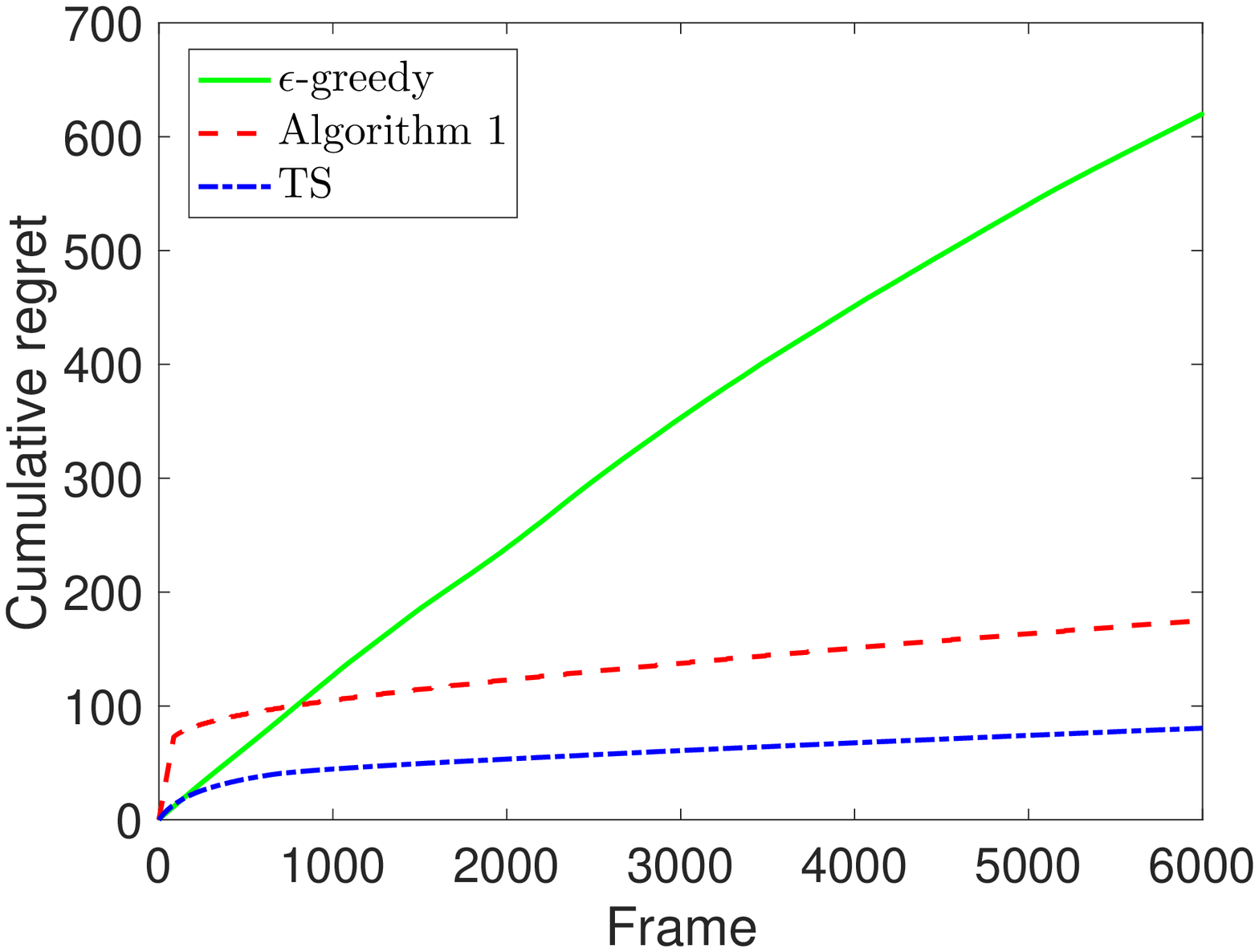}}\label{fig:2a} 
%	\quad
	\subfigure[Average regret]{ \includegraphics[width=0.4\textwidth, scale=1]{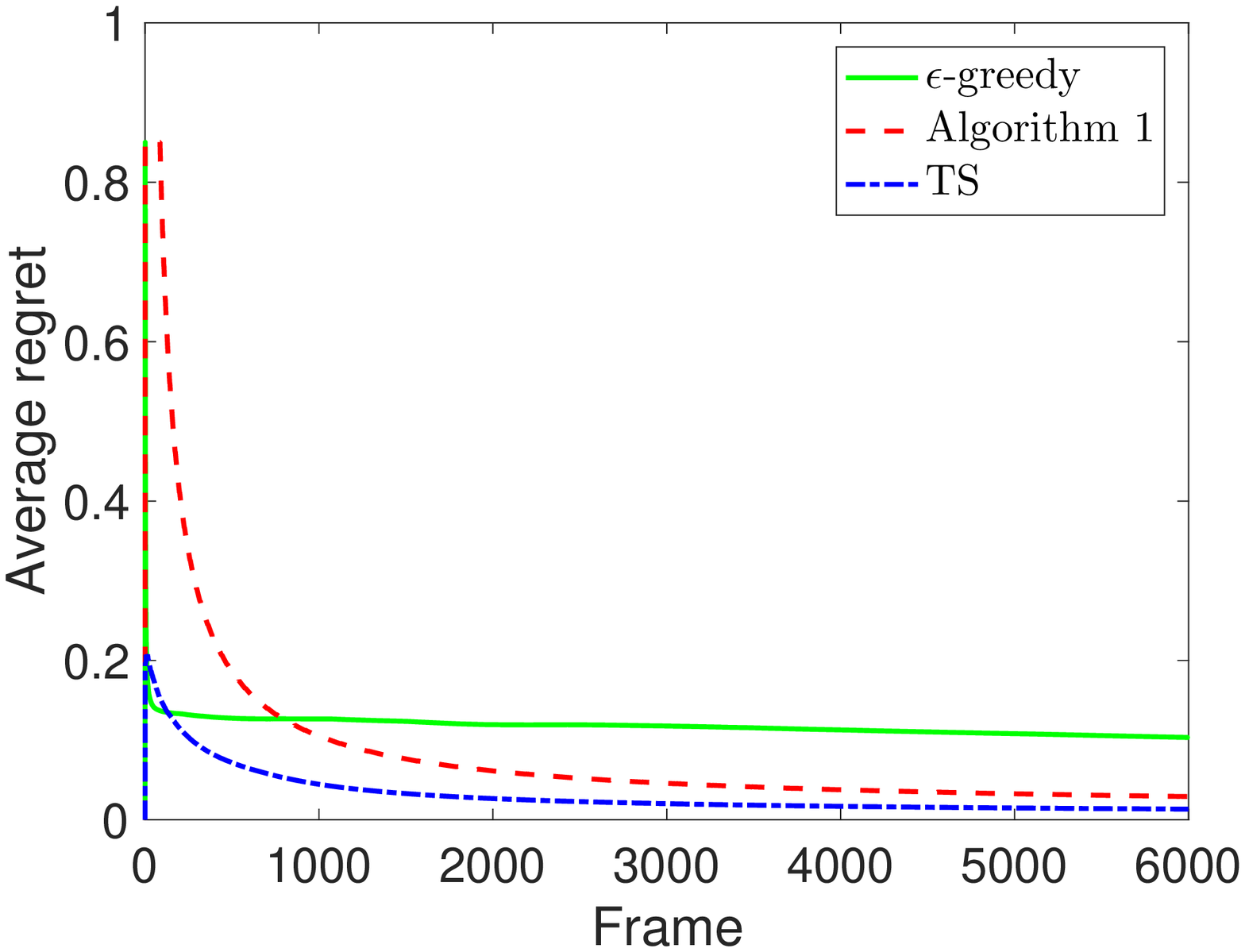}}\label{fig:2b}
%	\vskip\baselineskip
	\subfigure[Per-frame net reward]{ \includegraphics[width=0.4\textwidth, scale=1]{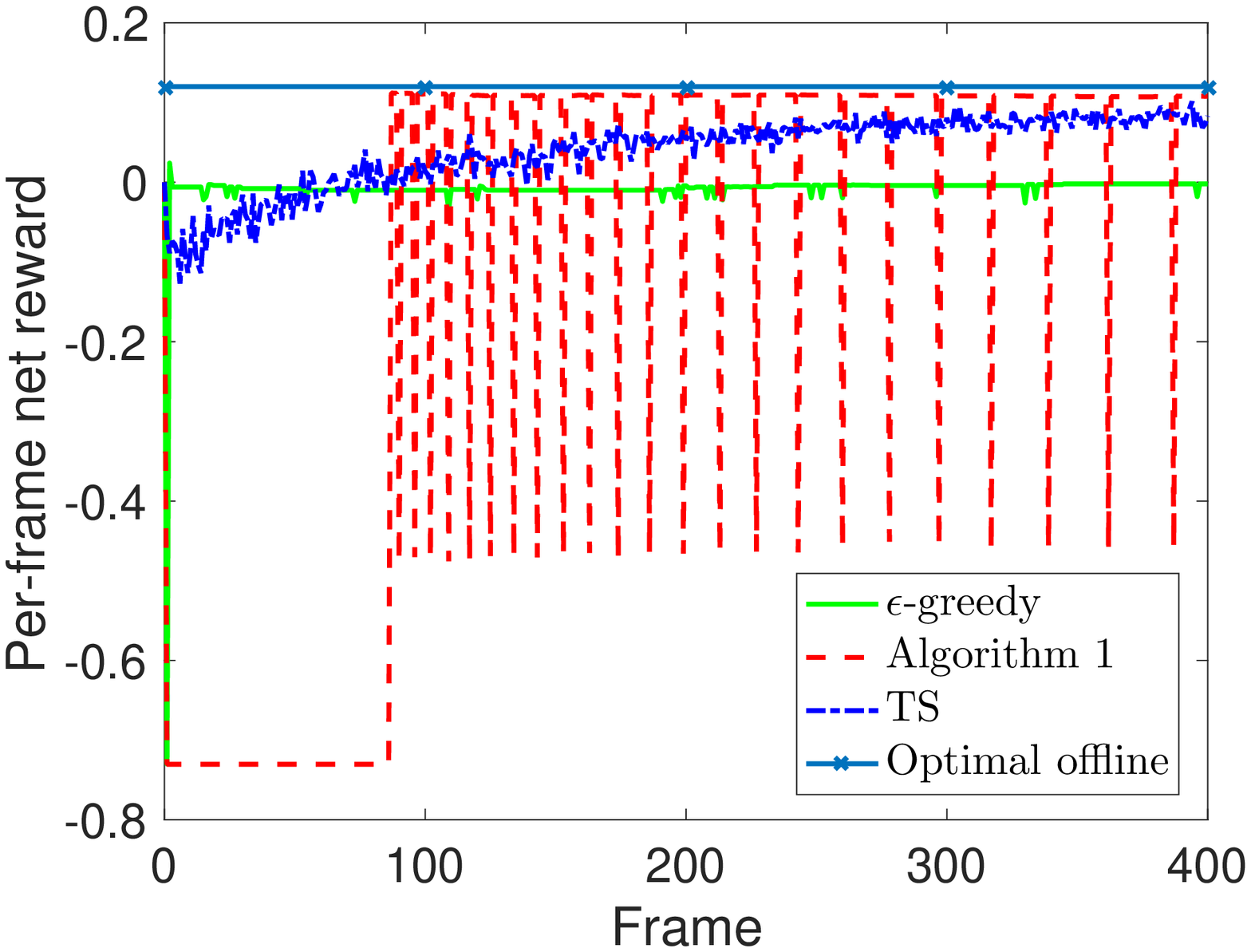}}\label{fig:3}
%	\quad
	\subfigure[Regret with different $L$]{ \includegraphics[width=0.4\textwidth, scale=1]{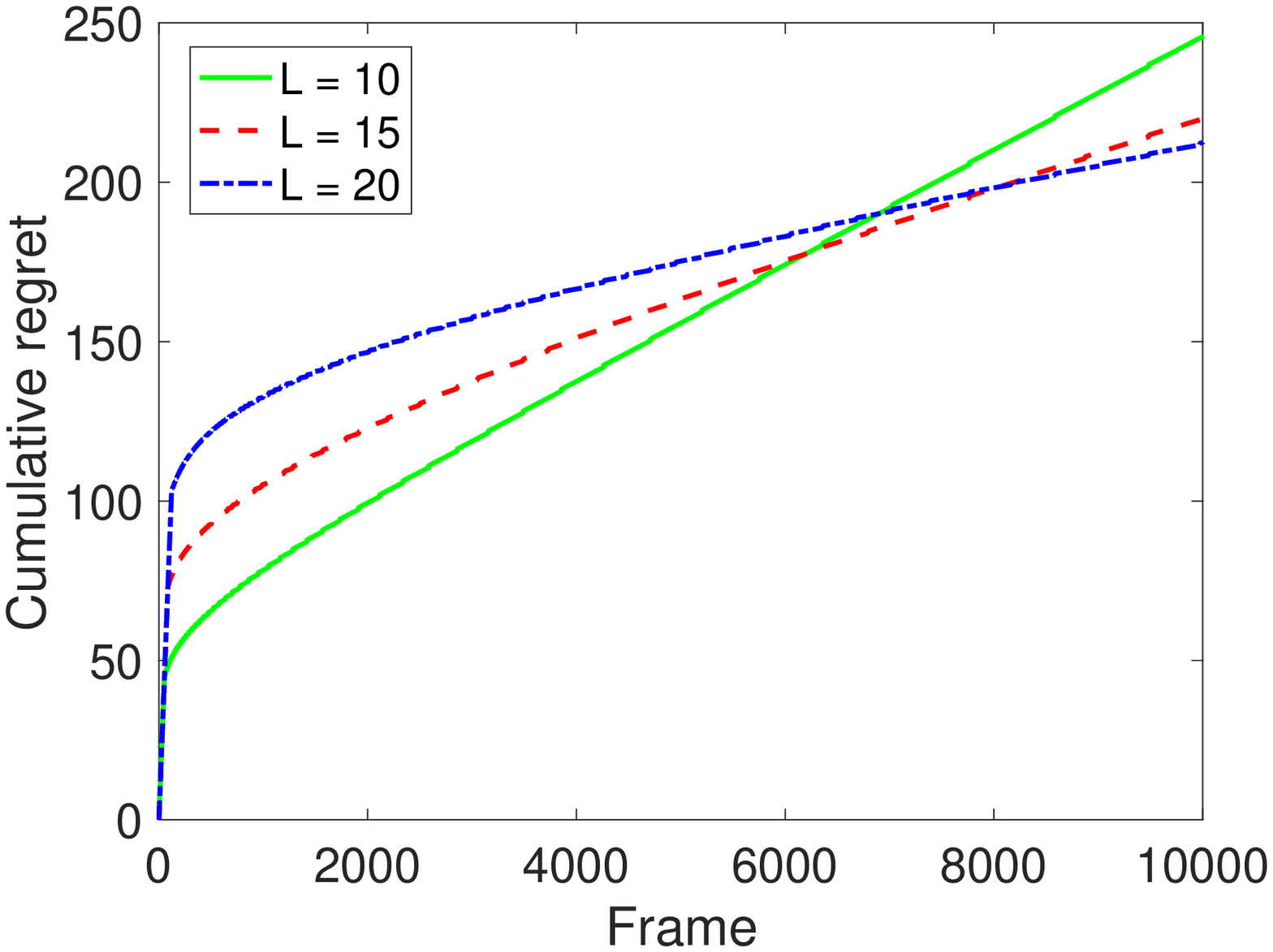}}\label{fig:4}
%	\vspace{-0.1in}
	\caption{Performance comparison.}\label{fig:2}
	%\vspace{-0.1in}
\end{figure*}

Denote $kl(\theta_1,\theta_2)$ as the KL-distance between two Bernoulli distributions with parameters $\theta_1$, $\theta_2$, respectively. Then, we have the following lower bound for $\alpha$-consistent strategies.
\begin{Theorem}\label{thm:lower_bound}
	Assume the optimal offline policy with $\theta_1\geq \theta_2\geq \ldots \geq \theta_K$ is to sense/guess the first $K^*$ channels sequentially according to Theorem~\ref{thm:offline} and quit the remaining channels, where $K^*<K$. 
	Then, under any $\alpha$-consistent strategy, we have
	\begin{align}
	\liminf_{T \rightarrow \infty}\frac{R(T)}{\log T} \geq \max_{k = K^*+1,\ldots,K} \min_{\pi \in \Pi_k} \frac{\Delta_\pi}{kl(\theta_k;\max(\theta_1,u'_k))}
	\end{align}
	where $\Pi_k$ is the set of policies under which channel $k$ might be sampled or used with guess, $\Delta_\pi$ is the per-frame regret under policy $\pi$, and $u'_k$ is the upper threshold for channel $k$ under the optimal offline policy, assuming $k$ is the best channel.
\end{Theorem}
Theorem~\ref{thm:lower_bound} indicates that for most system settings where at least one channel is left out under the optimal offline policy (i.e., $K^*<K$), the regret lower bound under any $\alpha$-consistent online strategy scales in $\log T$. This scaling matches the upper bound in Theorem~\ref{thm2}, thus our online strategy is order-optimal.  We restrict to the situation $K^*<K$ because, under this setting, we are able to identify some sub-optimal policies adopted under an online strategy easily, i.e., any policy involves sense or guess over channel $k$, $k>K^*$,  is strictly sub-optimal. By showing that the probability of choosing a policy that involves channel $k$ cannot be too small, we obtain a lower bound on the regret. The proof of Theorem~\ref{thm:lower_bound} can be found in Appendix~\ref{appx:thm3}.

\if{0}
\section{Time-varying Probing and Transmission Costs}\label{sec:time-varying}
In this section, we extend our analysis and algorithm in Sections~\ref{sec:offline} and \ref{sec:online} to the scenario where the probing and transmission costs are time-varying and known to the user before it makes any probing and transmission decisions. This can be interpreted as the case where the cost is subject to the resources available, and is time-varying if the resources available change in time. This can also correspond to the case where the cost reflects the QoS requirements for different applications, and adaptively changes if request from different applications arrive to the transceiver randomly.

To make the problem tractable, we make the following assumptions.
\begin{Assumptions}
\end{Assumptions}

In the above, we consider the situation that the reward $b_{0}$, probing cost $c_{0}$ and transmission cost $p_{0}$ are known and fixed. In reality, this may not be the case. Therefore, we have solutions for the following situations and give a brief description about the solutions.\\

Inspired by \cite{HW17}, we assume that the probing cost is time varying, which means the cost is random variable but we know it before we make any decision. We assume that the probing cost for all channels are the same and at time t is $c_{t}$, for simplicity we assume that $c_{t}$ can only take two values, which denote as $\{c_{1},c_{2}\}$, we have $\Pb[c_{t}=c_{1}]=\rho_{1},~\Pb[c_{t}=c_{2}]=\rho_{2}=1-\rho_{1} $ and  $c_{1}>c_{2}\geq0$. 

\subsubsection{Best Offline Policy}
Under the offline situation, i.e, we know all the statistics ($\theta_{i}, i\in[K]$). For fixed all statistic, the best strategy is fix (or time invariant).  We denote the best policy for probing cost equals to c as $G(c)$, as we can see, the best policy in the main part (specifically, the Theorem \ref{thm:offline}) is denoted as $G(c_{0})$. Then, in the time-varying case, the best policy for time t, $c_{t}=c_{1}$ can be denoted $G(c_{1})$ for time t, $c_{t}=c_{2}$ can be denoted $G(c_{2})$. Since, we didn't give any restriction about $c_{0}$ except $c_{0}>0$, then we can simply change the parameter $c_{0}$ in the Theorem \ref{thm:offline} and get the best policy $G(c_{1}),G(c_{2})$ for the cost $c_{1}, c_{2}$.

\subsubsection{Online Policy and Analysis}
Inspired by \cite{HW17}, we want to explore more aggressively to learn all statistics when the probing cost is low and explore conservatively when the probing cost is high. In other words, we have the following algorithm. Here, $T_{i}^{1}$ , ($T_{i}^{2}$), denote the number of prob of channel i under exploration and when $c_{t}=c_{1}$($c_{t}=c_{2}$)
\begin{algorithm}[H]
	\caption{Joint Learning and Optimization\\{\it Input}: $b_0$, $p_0$, $c_0$, $D(t)$.}
	\begin{algorithmic}[1]
		\State {\bf Initialization}: Sample each channel at $t=1$; Set $T_i (1) = 1$ for all $i\in[K]$; Record $\Fc_1$.%=\cup_{j}\{X_j(k)\}_{k=1}^{T_j(t)}$.
		\While {$t$}
		\State $t:=t+1$;
		\If $c_{t}=c_{1}$
		\State Let $\Ec(t):=\{i:T_i^{1}(t)<D_{c1}(t)\}$.
		\If {$\Ec(t)\neq \emptyset$} \Comment{{\it Exploration}}
		\State Sample every channel in $\Ec(t)$.
		\State Transmit over one available channel in $\Ec(t)$.
		\Else  \Comment{{\it Exploitation}}
		\State  Sort channels with $\hat{\theta}_i:=\frac{\sum_{k=1}^{T_i^{1}(t)}X_i(k)}{T_i^{1}(t)}$
		\State Calculate $\{u_i\}$, $\{l_i\}$ in Theorem~\ref{thm:offline} using $\{\hat{\theta}_i\}$.
		\State Take actions according to Theorem~\ref{thm:offline} using $\{\hat{\theta}_i\}$.
		\EndIf
		\Else 
		\State Let $\Ec(t):=\{i:T_i^{2}(t)<D_{c2}(t)\}$.
		\If {$\Ec(t)\neq \emptyset$} \Comment{{\it Exploration}}
		\State Sample every channel in $\Ec(t)$.
		\State Transmit over one available channel in $\Ec(t)$.
		\Else  \Comment{{\it Exploitation}}
		\State  Sort channels with $\hat{\theta}_i:=\frac{\sum_{k=1}^{T_i^{2}(t)}X_i(k)}{T_i^{2}(t)}$
		\State Calculate $\{u_i\}$, $\{l_i\}$ in Theorem~\ref{thm:offline} using $\{\hat{\theta}_i\}$.
		\State Take actions according to Theorem~\ref{thm:offline} using $\{\hat{\theta}_i\}$.
		\EndIf
		\EndIf
		\State Update $T_i^{1}(t)$, $T_i^{2}(t)$,$\Fc_t$.
		\EndWhile
	\end{algorithmic}\label{alg:online3}
\end{algorithm}

\begin{Theorem}\label{them:s2}
	There exist sufficiently large constant $L(c_{1}), D(c_{1}), L(c_{2}), D(c_{2}) $ such that the regret under Algorithm\ref{alg:online3} is bounded by \\
	$$R(T)\leq C_{1}\log(\rho_{1}T)+D_{1}
	+C_{2}\log((1-\rho_{1})T)+D_{2} $$
	where $C_{1}=L(c_{1})(b_{0}+Kc_{1}), D_{1}=\frac{\pi^2}{6}+D(c_{1})$\\ $C_{2}=L(c_{2})(b_{0}+Kc_{2}), D_{2}=(\frac{\pi^2}{6}+D(c_{1}))$
\end{Theorem}
\begin{proof}
	We simply partition the time to be two part $\{1,2,\ldots,T\}=S_{1}\bigcup S_{2}$ and $S_{1}\bigcap S_{2}=\emptyset$, the first part $S_{1}$ is set of time when $c_{t}=c_{1}$, the second part $S_{2}$ is the set of time when $c_{t}=c_{2}$. Based on the analysis in the main part, we can divide the regret $R(T)$ into two part $R_{1}(|S_{1}|)$ and $R_{2}(|S_{2}|)$, $R_{1}(|S_{1}|)$ is the expectation of regret when $t \in S_{1}$ , $R_{2}(|S_{2}|)$ is the expectation of regret when $t \in S_{2}$
	\begin{eqnarray}
	&& R(T) \\
	&\leq& \Eb_{|S_{1}|}[R_{1}(|S_{1}|)+R_{2}(|S_{2}|)]\\
	&=&\Eb_{|S_{1}|}[R_{1}(|S_{1}|)]+\Eb_{|S_{1}|}[R_{2}(|S_{2}|)]\\
	&=&\Eb_{|S_{1}|}[R_{1}(|S_{1}|)]+\Eb_{|S_{2}|}[R_{2}(|S_{2}|)] \label{eq1}\\
	&\leq&\Eb_{|S_{1}|}[C_{1,c_{1}}\log|S_{1}|+C_{2,c_{1}}]\\ &+&\Eb_{|S_{2}|}[C_{1,c_{2}}\log|S_{2}|+C_{2,c_{2}}]\\
	&\leq& C_{1,c_{1}}\log\Eb_{|S_{1}|}(|S_{1}|)+C_{2,c_{1}}\\
	&+& C_{1,c_{2}}\log\Eb_{|S_{2}|}(|S_{2}|)+C_{2,c_{2}}\label{eq2}\\
	&=& C_{1,c_{1}}\log(\rho_{1}T)+C_{2,c_{1}}\\
	&+& C_{1,c_{2}}\log((1-\rho_{1})T)+C_{2,c_{2}}
	\end{eqnarray}
	here $C_{1,c_{i}}=L(c_{i})(b_{0}+Kc_{i}),C_{2,c_{i}}=(\frac{\pi^2}{6}+D(c_{i}))(b_{0}+Kc_{i},~ \forall i\in \{1,2\})$\\
	The equality (\ref{eq1}) comes from simple fact that $|S_{1}|+|S_{2}|=T$, T is fixed at this point, therefore, take expectation with $|S_{1}|$ is equivalent to take expectation with $|S_{2}|$. \\
	The inequality (\ref{eq2}) comes from the concavity of function $\log$
\end{proof}

As we can see, the expectation of regret of time varying case is on the scale of $O(\log(\rho_{1}T)+\log((1-\rho_{1})T))$, which is significantly smaller than the scale in Theorem $O(\log(t))$. In other words, by incorporating the knowledge of changing $c_{t}$, we can reduce the regret significantly.

\fi

\section{Simulation Results}\label{sec:simulation}
In this section, we evaluate the offline and online spectrum access policies through numerical results.
\subsection{Offline Policy}
We first study how of system parameters would affect the optimal offline policy. We set ${\thetav} = \{0.6,0.5,0.4,0.3,0.2,0.1\}$, $b_0 = 1$, and change the values of $p_0$ and $c_0$ separately. As shown in Table~\ref{table:offline}, for fixed $p_0$, the maximum number of channels to sense or access (i.e., $N$) is monotonically decreasing as $c_0$ increases. However, when $c_0$ is fixed, we do not observe such monotonicity in $p_0$. This can be explained as follows: When $p_0$ is small, the cost of a wrong {\it guess} is small compared with the cost of sensing, thus the system would choose to transmit over the best channel directly without sensing.  When $p_0$ is large, the cost of a wrong {\it guess} outweighs that of sensing, thus the user should choose to sense instead of guess. As $p_0$ grows, the user will probe less channels because the potential gain by sensing a channel (i.e., $\theta_i(b_0-p_0)$) may not cover the sensing cost. 

\begin{table}[t]
	\begin{tabular}{c c c c } 
		\toprule[1.5pt]
		$p_0$ & $c_0$ & $N$ & Action over channel $N$\\ 
		\midrule
		0.50 & 0.15 & 4 &  sense  \\ 
		\hline
		0.50 & 0.17 & 3 &  sense \\ 
		\hline
		0.50 & 0.21 & 2 & sense \\ 
		\hline
		0.50 & 0.23 & 1 & guess\\
		\hline
		0.30 & 0.20 & 1 & guess \\ 
		\hline
		0.40 & 0.20 & 3 & sense \\ 
		\hline
		0.60 & 0.20 & 2 & sense \\ 
		\hline
		0.65 & 0.20 & 1 & sense \\
		\bottomrule[1.5pt]
	\end{tabular}
%\vspace{-0.1in}
	\centering
	\caption{Structure of the optimal offline policy.}\label{table:offline}
%	\vspace{-0.1in}
\end{table}

\subsection{Online Algorithm}
Then, we resort to numerical simulations to verify the effectiveness of Algorithm \ref{alg:online} in the online learning situation. 

Before we present the simulation results, we first introduce two baseline learning algorithms for comparison, namely,
the $\epsilon$-greedy learning algorithm, and the Thompson Sampling (TS) based algorithm.
The only difference between the $\epsilon$-greedy algorithm and Algorithm~\ref{alg:online} is to replace the condition for exploration with $S_t = 1$, where $S_t$ is an independent sample of a Bernoulli distribution with parameter $\epsilon$. Thus, $\epsilon$-greedy algorithm explores the system at a fixed rate. 
TS is a randomized algorithm based on the Bayesian principle, and has generated significant interest due to its superior empirical performance~\cite{Shipra:TS:2012}. We tailor the TS algorithm to our setting, where the key idea of the algorithm is to sample the channel statistics according to Beta distributions, whose parameters are determined by past observations. Though we are not able to characterize its performance theoretically, we do observe significant performance improvement through simulations. Therefore, we include it in this section for comparison purpose. {\em Obtaining an upper bound on the regret under the TS algorithm is one of our future directions.}

%Note that the $\epsilon$-greedy algorithm explores the system at a fixed rate, and the numerical comparison will shed light on how algorithm~\ref{alg:online} benefits from balancing exploration and exploitation while learning from the system. {\it Thompson-sampling} algorithm, on the other hand, is best known for its peerless good empirical performance in many problems related to MAB, and the numerical comparison is meant to demonstrate the gap between our algorithm and the Thompson Sampling algorithm and motivate future works.
% still give this as a better solution to deal with the exploration-exploitation dilemma with numerical results. 
\if{0}
\begin{algorithm}[t]
	\caption{Thompson Sampling based Learning\\{\it Input}: $b_0$, $p_0$, $c_0$.}
	\begin{algorithmic}[1]
		\State {\bf Initialization}: Set $t=0$, $S_k = F_k = 0$ for each $k \in [K]$.
		\While {$t$}
		\State $t:=t+1$;
		\State Sample $\hat{\theta}_k$ according to $Beta(S_k+1, F_k+1)$ for each $k \in [K]$.
		\State Calculate the thresholds in Theorem~\ref{thm:offline} using $\hat{\theta}_k$.
		\State Make decisions according to Theorem~\ref{thm:offline} using $\hat{\theta}_k$.
		\State Record the observed channel status in $\Fc_t$.
		\State For each observed channel $j$, set $$(S_j, F_j) = (S_j, F_j) +\lv_{X_j = 1}(1,0) +  \lv_{X_j = 0}(0,1).$$
		\EndWhile
	\end{algorithmic}\label{alg:thmps}
\end{algorithm}
%	\vspace{-0.1in}
\fi	

We now compare the performances under Algorithm~\ref{alg:online}, the $\epsilon$-greedy algorithm, and the TS algorithm through simulations. We set ${\thetav} = \{0.6,0.5,0.4,0.3,0.2,0.1\}$, $b_0 = 1$, $p_0 = 0.5$, $c_0 = 0.2$, and let the sensing cost, transmission cost and reward be uniform random variables with $\Delta_b=\Delta_p=\Delta_c=0.1$. According to Theorem \ref{thm:offline}, the optimal offline policy is to sense the first three channels sequentially until it finds the first available channel, and quit the current frame if all of them are unavailable. The expected per-frame net reward under the optimal policy is $0.12$. We first set $L=20$, $D=24.85$ for Algorithm~\ref{alg:online}, and $\epsilon=0.001$ for the $\epsilon$-greedy algorithm.
We run each algorithm 100 times and plot the sample-path average in Fig.~\ref{fig:2}(a) and Fig.~\ref{fig:2}(b). As we can see in Fig.~\ref{fig:2}(a), Algorithm~\ref{alg:online} and TS algorithm outperform the $\epsilon$-greedy algorithm significantly as time $T$ becomes sufficiently large. In addition, the $\epsilon$-greedy algorithm achieves smaller regret than Algorithm~\ref{alg:online} when $T$ is small. This is because Algorithm~\ref{alg:online} explores more aggressively initially, resulting in a larger regret. In Fig.~\ref{fig:2}(b), we also notice that the average regrets of Algorithm~\ref{alg:online} and TS algorithm approach zero as $T$ increases, indicating that both algorithms perform better after gaining information of the system. In contrast, the average regret under the $\epsilon$-greedy algorithm approaches zero quickly and never converges to zero, which implies that it does not balance exploration and exploitation properly.

Next, we evaluate the per-frame net reward under those three algorithms, and compare them against the per-frame net reward under the optimal offline policy in Fig.~\ref{fig:2}(c). We notice that the per-frame net rewards under both Algorithm~\ref{alg:online} and TS algorithm converge to the upper bound. The per-frame net reward under Algorithm~\ref{alg:online} drops significantly in certain frames, as indicated by the sharp spikes. This is because Algorithm~\ref{alg:online} has separated exploration stages. Whenever some of the channels have not been observed for sufficient number of times, the user will enter the exploration stage. We also note that the interval between two consecutive spikes grows as $T$ increases, indicating the portion of the exploration stage decays in time.

Finally, we focus on the impact of parameter selection on the performance under Algorithm~\ref{alg:online} in Fig.~\ref{fig:2}(d). We set $L$ to be 10 , 15, 20 respectively while keeping $D=\frac{\log (2K)}{2}L$ according to Lemma~\ref{summary} in Appendix~\ref{appx:online}. Interestingly, we note that the performance does not change monotonically as $L$ varies. Specifically, when $L$ is 10, it results in the smallest regret initially. However, as $T$ increases, the regret becomes even larger than those with $L=15$ and 20. It can be interpreted in this way: the algorithm explores less with a smaller $L$, thus saving the sensing cost in exploration. However, it also converges to the optimal offline policy at a lower rate, as it has less observations. In contrast, the algorithm explores more aggressively with a larger $L$, leading to larger cost in exploration but a faster convergence rate. Initially, the cost of exploration outweighs the reward of transmission, thus a smaller regret can be observed for smaller $L$. When $T$ grows, the reward of transmission outweighs the cost of exploration, and the regret is mainly determined by the estimation accuracy. Therefore, the regret with smaller $L$ grows faster as $T$ increases, and eventually becomes greater than that with larger $L$.

\section{Discussions and Conclusions}\label{sec:discussion}
In this paper, we investigated a discrete-time multi-channel cognitive radio system with random sensing and transmission costs and reward. We started with an offline setting and explicitly identified the recursive double-threshold structure of the optimal solution. With insight drawn from the optimal offline policy, we then studied the online setting and proposed a order-optimal online learning and optimization algorithm. We further compared our algorithm with other baseline algorithms through simulations.

The problem studied in this paper is essentially a cascading bandit problem with ``soft" cost constraint, which itself is a non-trivial extension of \cite{KvetonSWA15}. We believe that the design and analysis of the algorithms in this paper advances the  state of the art in both cognitive radio systems and MAB, and has the potential to impact a broader class of cost-aware learning and optimization problems.

\appendix

\subsection{Proof of Lemma~\ref{lemma:ranking}}\label{appx:lemma:ranking}
We first prove the first half of Lemma~\ref{lemma:ranking} through contradiction. Assume that under the optimal policy, the transmitter senses a channel $i$ right ahead of channel $j$, where $\theta_i<\theta_j$. We construct an alternative policy by switching the probing order of $i$ and $j$. Consider a fixed realization of all $X_i$, $i\in[K]$. Then, those two policies would result in different actions only when all channels sampled ahead of channel $i$ are unavailable, and $X_i\neq X_j$.

Case I: $X_i=1, X_j=0$. This event happens with probability $\theta_i(1-\theta_j)$.
Under the original policy, the user would transmit over channel $i$ after probing. The instantaneous net reward in this step would be $b_0-p_0-c_0$. After switching, the user will transmit over channel $i$ after probing channel $j$ and then channel $i$. This results additional probing cost $c_0$ with the same reward.

Case II: $X_i=0, X_j=1$. This event happens with probability $\theta_j(1-\theta_i)$.
Under the original policy, the user would transmit over channel $j$ after probing both channels $i$ and $j$. After switching, the user will transmit over channel $j$ after probing, but not probing channel $i$. This probing cost will be $c_0$ less with the same reward.

Thus, the expected difference in probing cost would be
\begin{align}
\theta_i(1-\theta_j)c_0+\theta_j(1-\theta_i)(-c_0)=(\theta_i-\theta_j)c_0<0
\end{align}
Therefore, by switching the probing order of channel $i$ and channel $j$, we save probing cost without reducing the reward in expectation. Thus, under the optimal policy, the transmitter must sense the channels sequentially according to a descending order of their means.

We then prove the second half of Lemma~\ref{lemma:ranking} through contradiction as well. Assume the worst channel in $\Pc$, denoted as $i$, is worse than the best channel in $[K]\backslash \Pc$, denoted as $j$ (i.e., $\theta_i<\theta_j$). Under the original policy, there might be two possible actions after probing channel $i$: {\it guess}, i.e., transmit over $j$ without probing, or {\it quit}.  We construct an alternative policy by switching the role of $i$ and $j$. Consider a fixed realization of all $X_i$, $i\in[K]$. Again, those two policies would result in different actions only when all channels sampled ahead of channel $i$ are unavailable, and $X_i\neq X_j$.

Case I: $X_i=1, X_j=0$.
Under the original policy, the user would transmit over channel $i$ after probing. After switching, depending on the action on channel $j$ under the original policy, the user will first sense channel $j$, and then {\it guess} or {\it quit}. This will not incur any additional probing cost, however, the corresponding reward minus transmission cost would be $b_0-p_0$ or $0$.

Case II: $X_i=0, X_j=1$.
Under the original policy, the user would {\it guess} or {\it quit} after probing channel $i$. The corresponding reward would be $b_0-p_0$ or $0$. After switching, the user will transmit over channel $j$ after probing. Again, this will not incur any additional probing cost, however, the corresponding reward minus transmission cost would be $b_0-p_0$.

In conclusion, if the original policy is to {\it guess} after probing channel $i$, there will be no difference in reward and cost after switching for both cases; If the original policy is to {\it quit} after probing channel $i$, then the difference in reward minus transmission cost would be
\begin{equation}
-\theta_j(1-\theta_i)(b_0-p_0)+\theta_i(1-\theta_j)(b_0-p_0) 
=(\theta_i-\theta_j)(b_0-p_0)>0
\end{equation}
Therefore, by switching the role of channel $i$ and channel $j$, we increase net reward in expectation. Thus, under the optimal policy, any channel in $\Pc$ should be better than any other channel outside $\Pc$.

\subsection{Proof of Theorem 1}\label{appx:thm1}
%In this section, we provide the proof of Theorem~\ref{thm:offline}.
%Our proof consists of three main parts. First we proof a simple fact. Then,we prove we should examine channels sequentially from channel 1*. After that  we prove we should only take one of three actions and there exists threshold $a_{j}, b_{j}, j\in [K] $, we are supposed to guess, prob or quit based on the relationship between the threshold and $\theta_{j}$, and we express $a_{j}, b_{j}$ .

We prove the theorem through induction. As shown in Lemma~\ref{lemma:probe}, if $X_i=1$, then the transmitter should transmit over channel $i$. If $X_i=0$, the transmitter need to decide between three actions: continue probing the next best channel, transmit over the next best channel without probing, or quit.

Define $\Sc_i:=[i,\ldots, K], \quad \forall i=1,2,\ldots, K$.
%\begin{align}\label{eqn:S_i}
%\Sc_i&:=[i,\ldots, K], \quad \forall i=1,2,\ldots, K
%\end{align}
Then, based on Lemma~\ref{lemma:ranking} and the optimality condition in (\ref{eqn:bellman}), we have
\begin{align}
V(i, 1,\Sc_{i+1})&=b_0-p_0, \quad \forall i=1,2,\ldots,K-1,\\
V(i, 0,\Sc_{i+1})&=\max\big\{ -c_0 +\Eb[V(i+1, X_{i+1},\Sc_{i+2})],\nonumber\\
&\qquad\theta_{i}b_0- p_0, \,0 \big\}, \quad \forall i=0,\ldots,K-1,\\
V(K, X_{K},\emptyset)&=\max\left\{ X_{K} b_0- p_0,\,0 \right\}.\label{eqn:V_K}
\end{align}
Therefore,
\begin{align}\label{eqn:expected_V}
\Eb[V(i, X_{i},\Sc_{i+1})]\hspace{-0.02in}&=\hspace{-0.02in}(b_0-p_0)\theta_{i}\hspace{-0.02in}+\hspace{-0.02in}V(i, 0,\Sc_{i+1})(1-\theta_{i}).
\end{align}
Denote
\begin{align}
E_i&:=V(i, 0,\Sc_{i+1}),\quad \forall i=0,1,\ldots, K.
\end{align}
Then, for  $i=1,2,\ldots,K$,
\begin{align}\label{eqn:recursive}
E_{i-1}&=\max\left\{ -c_0 +(b_0-p_0)\theta_{i}+E_i(1-\theta_{i}), \theta_{i}b_0- p_0, 0 \right\}.
\end{align}
We note that the first two terms inside the max operator in (\ref{eqn:recursive}) are linear functions of $\theta_{i}$. By plotting them in Fig.~\ref{fig:linear}, we observe that $E_{i-1}$ is a piecewise linear function in $\theta_{i}$. Specifically, define
\begin{align}\label{eqn:a_i}
u_{i} &=\max\left\{\frac{p_{0}}{b_{0}}, 1-\frac{c_{0}}{p_{0}+E_i}\right\},\\
l_{i} &=\min\left\{\frac{p_{0}}{b_{0}}, 1-\frac{b_0-p_0-c_{0}}{b_{0}-p_{0}-E_i}\right\}.\label{eqn:b_i}
\end{align}
Then,
\begin{align}\label{eqn:E_i}
E_{i-1}&=\left\{ \begin{array}{ll}
0, & \theta_{i}\in [0, l_i]\\
-c_0 +(b_0-p_0)\theta_{i}+E_i(1-\theta_{i}), & \theta_{i}\in [l_i, u_i]\\
\theta_{i}b_0- p_0 & \theta_{i}\in[u_i,1 ]
\end{array}\right.
\end{align}
and the corresponding actions that the user should take after finding $X_{i-1}=0$ are {\it quit}, {\it sense}, and {\it guess}, respectively.

\begin{figure} [t]
	\centering
	\subfigure[$E_i$ is large ]{ \includegraphics[width=0.225\textwidth]{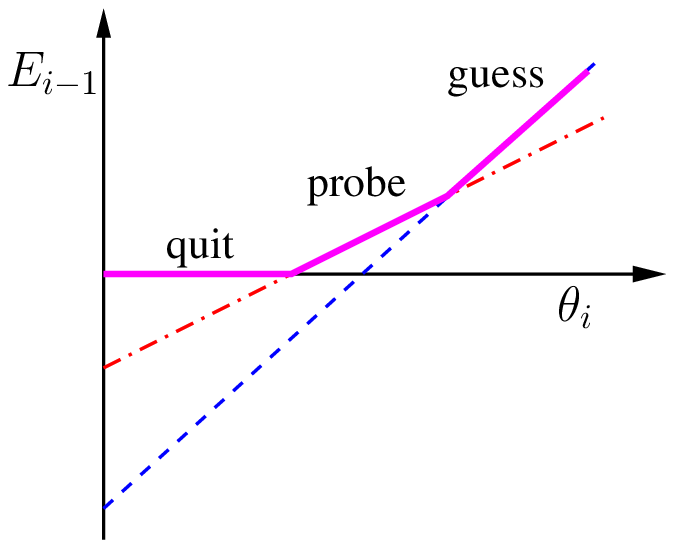}}
	\subfigure[$E_i$ is small ]{ \includegraphics[width=0.225\textwidth]{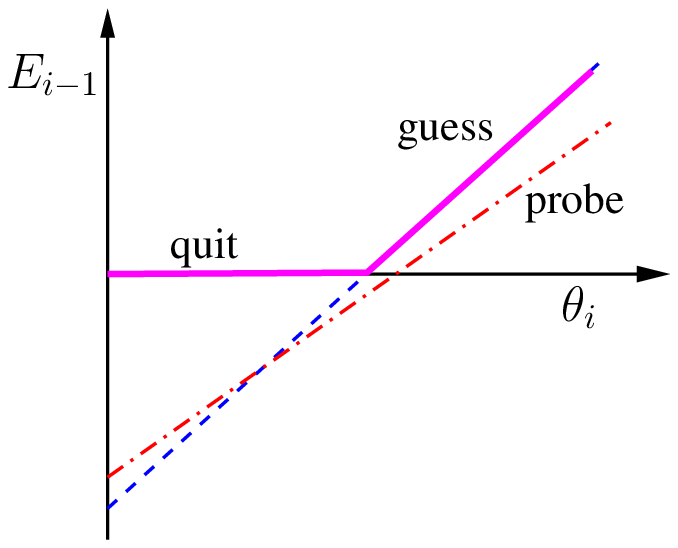}}
	\caption{$E_{i-1}$ as a function of $\theta_i$.}
	%\vspace{-0.1in}
	\label{fig:linear}
\end{figure}

\subsection{Proof of Theorem 2}\label{appx:online}
The proof consists of three main steps: we begin by first relate the difference between the thresholds $l_{i}, u_{i}$ and their estimates $\hat{l}_{i}, \hat{u}_{i}$ with $\{\theta_{i}-\hat{\theta}_i\}$, $c_0-\hat{c}_0$, $b_0-\hat{b}_0$, $p_0-\hat{p}_0$. Then, we derive an upper bound on the number of samples required to ensure the correct ordering of the channels, as well as the right sensing and transmission decisions with high probability. Finally, we explicitly bound the regret by examining the regret from exploration and exploitation separately.

To facilitate our analysis, we first introduce the following two lemmas without proof.
\begin{Lemma}\label{lemma:max}
	Let $A=\max\{a_1,a_2,\ldots,a_n\}$, $B=\max\{b_1,b_2,\ldots,b_n\}$, where $a_i,b_i\in \Rb$. Then, $$|A-B|\leq \max\{|a_1-b_1|, |a_2-b_2|,\ldots, |a_n-b_n|\}.$$
\end{Lemma}

\begin{Lemma}\label{lemma:E_i}
	$E_i$ monotonically decreases as $i$ increases from $0$ to $K$.
\end{Lemma}	

\if{0}
\begin{proof}
	We prove this lemma through induction.
	According to (\ref{eqn:V_K})(\ref{eqn:expected_V}) and (\ref{eqn:recursive}), we have $E_K=0$, and $E_{K-1}=\max\{-c_0 +(b_0-p_0)\theta_{K}, \theta_{K}b_0- p_0, 0 \}$, which is bounded between $0$ and $b_0-p_0$. Thus, we have $b_0-p_0\geq E_{K-1}\geq E_K$.
	Assume $b_0-p_0\geq E_{i-1}\geq E_i$ is true for $i\leq K$. Then, we need to show that
	$b_0-p_0\geq E_{i-2}\geq E_{i-1}$.
	
	Based on the recursive relationship in (\ref{eqn:recursive}), it suffices to show
	\begin{align}
	-c_0 +(b_0-p_0)\theta_{i-1}+E_{i-1}(1-\theta_{i-1}) \nonumber \\
	\geq  -c_0 +(b_0-p_0)\theta_{i}+E_{i}(1-\theta_i)
	\end{align}
	Since we assume $E_{i-1}\geq E_{i}$, it then suffices to prove that
	\begin{align}
	(b_0-p_0)\left(\theta_{i-1}-\theta_{i}\right)+E_{i}(1-\theta_{i-1})\geq E_{i}(1-\theta_i)
	\end{align}
	which is true due to the assumption that $b_0-p_0\geq E_i$.
\end{proof}
\fi

\begin{Lemma}\label{lemma:estimate_error_bound}
	If for a sufficiently small $\epsilon$, we have $|\hat{\theta}_{i}-\theta_{i}|<\epsilon$, $\forall i$, $|\hat{b}_{0}-b_{0}|<\epsilon$, $|\hat{c}_{0}-c_{0}|<\epsilon$, $|\hat{p}_{0}-p_{0}|<\epsilon$, then there exist constants $A_{i}, B_{i}$, such that $|\hat{u}_{i}-u_{i}|< A_{i}\epsilon$, $|\hat{l}_{i}-l_{i}|< B_{i}\epsilon$, $\forall i$.
\end{Lemma}

\begin{proof}
	Based on (\ref{eqn:E_i}), the estimate of $E_{i}$, denote as $\hat{E}_{i}$, can be expressed as follows
	\begin{align*}
	\hat{E}_{i}=\max \{ & -\hat{c}_{0}+\hat{\theta}_{i+1}(\hat{b}_{0}-\hat{p}_{0})+(1-\hat{\theta}_{i+1})\hat{E}_{i+1}, \\
	&\quad\hat{\theta}_{i+1} \hat{b}_{0}-\hat{p}_{0},	\,0 \}.
	\end{align*}	
	Denote $\bar{E}_{i}:=|\hat{E_{i}}-E_{i}|$. According to Lemma~\ref{lemma:max}, we have
	\begin{align}
	\bar{E}_{i}\leq   \max\{ & |\hat{c}_{0}-c_{0}|+|\hat{\theta}_{i+1}\hat{b}_{0}-\theta_{i+1}b_{0}|+|\hat{\theta}_{i+1}\hat{p}_{0}-\theta_{i+1}p_{0}|\nonumber\\
	&+|\hat{E}_{i+1}-E_{i+1}|+|\hat{\theta}_{i+1}\hat{E}_{i+1}-\theta_{i+1}E_{i+1}|,\nonumber\\
	& |\hat{\theta}_{i+1}\hat{b}_{0}-\theta_{i+1}b_{0}|+|\hat{p}_{0}-p_{0}|,\,	0\}.	\label{eqn:Delta1}
	\end{align}

	We note that
	\begin{align*}
	&	|\theta_{i+1} E_{i+1}-\hat{\theta}_{i+1} \hat{E}_{i+1}|\nonumber \\
	&= | \theta_{i+1} E_{i+1}- \hat{\theta}_{i+1} E_{i+1} + \hat{\theta}_{i+1} E_{i+1}- \hat{\theta}_{i+1} \hat{E}_{i+1}  | \\
	&\leq | \hat{\theta}_{i}-\theta_{i}| E_{i+1}+\hat{\theta}_{i+1} |\hat{E}_{i}-E_{i}|\leq \epsilon (b_{0}-p_{0})+\bar{E}_{i+1}.
	\end{align*}	
	Besides, we have
	\begin{align}
	& |\theta_{i+1} p_{0}-\hat{\theta}_{i+1} \hat{p}_{0}|=|\theta_{i+1} p_{0}-\hat{\theta}_{i+1} {p}_{0}+\hat{\theta}_{i+1} {p}_{0}-\hat{\theta}_{i+1} \hat{p}_{0}|\nonumber\\
	&\leq | \hat{\theta}_{i+1}-\theta_{i+1}| p_{0}+ |\hat{p}_{0}-p_{0}|\hat{\theta}_{i+1}\leq \epsilon (1+p_0),\label{eqn:bound_e}
	\end{align}
	where (\ref{eqn:bound_e}) follows from the fact that $\hat{\theta}_{i+1}\leq 1$. 
	
	Similarly, we can show that $|\theta_{i+1} b_{0}-\hat{\theta}_{i+1} \hat{b}_{0}|<\epsilon (1+b_0)$.
	Plugging into (\ref{eqn:Delta1}), we have
	\begin{align}
	\bar{E}_{i}
	&\leq \max\{\epsilon+\epsilon(1+b_0)+\epsilon(1+p_0)+\bar{E}_{i+1}\nonumber\\
	&\qquad+\epsilon(b_{0}-p_{0})+\bar{E}_{i+1}, \epsilon(1+b_0)+\epsilon,\,	0\}\nonumber\\
	&\leq \epsilon+\epsilon(1+b_0)+\epsilon(1+p_0)+\bar{E}_{i+1}+\epsilon(b_{0}-p_{0})+\bar{E}_{i+1}\nonumber\\
	&=  \epsilon (3+2b_{0})+2 \bar{E}_{i+1}.\label{eqn:Delta_bound1}
	\end{align}
	Multiplying $2^i$ to both sides of (\ref{eqn:Delta_bound1}), we have
	\begin{align}
	2^i\bar{E}_{i}&\leq 2^i \epsilon (3+2 b_{0})+2^{i+1} \bar{E}_{i+1},\quad\forall i\in[K].\label{eqn:Delta_recurvise1}\nonumber
	\end{align}
	Applying (\ref{eqn:Delta_recurvise1}) recursively, we have
	\begin{align}
	2^i \bar{E}_{i}
	&\leq \sum_{j=i}^{K-1}2^j \epsilon (3+2b_{0})+2^K \bar{E}_{K}\\
	&\leq \epsilon (3+2 b_{0})(2^K-2^i)+2^K\epsilon(b_0+2), \label{eqn:geometric_series1}
	\end{align}
	where (\ref{eqn:geometric_series1}) follows from the fact that {$\hat{E}_{K}=\max\{\hat{\theta}_{k} \hat{b}_{0}-\hat{p}_{0},0\} $, thus $\bar{E}_{K}\leq \epsilon (b_{0}+2)$} accordingly to Lemma~\ref{lemma:max}.
	Therefore,
	\begin{align}\label{eqn:Delta_e1}
	\bar{E}_{i}&\leq \epsilon \frac{(3+2 b_{0})(2^K-2^i)+2^K (b_0+2)}{2^i} :=\epsilon  e_i,	
	\end{align}
	where 
	\begin{align}\label{eqn:e_i_g}
	e_{i}:=\frac{(3+2 b_{0})(2^K-2^i)+2^K (b_0+2)}{2^i}.
	\end{align}	
	
	Next, we will use the relationship between $\bar{E}_i$ and $e_i$ in (\ref{eqn:Delta_e1}) to bound $|l_{i}-\hat{l}_{i}|$ and $|u_{i}-\hat{u}_{i}|$.
	Toward this end,	we first note that for any $a,b,\hat{a},\hat{b}$ satisfying $|a-\hat{a}|<\epsilon_a$ and $|b-\hat{b}|<\epsilon_b$, 
	\begin{align}
	&|\hat{a}b-a\hat{b}| = |(\hat{a}-a)b+a(b-\hat{b})| \leq \epsilon_a|b|+\epsilon_b|a|.\label{fact}
	\end{align}
	%When $a>0,b>0$, we can simplify \ref{fact} by $|\hat{a}b-a\hat{b}\leq \epsilon(a+b)$
	Therefore, 
	\begin{align}
	&|u_{i}-\hat{u}_{i}|\nonumber\\
	&=\bigg|\max\left\{\frac{p_{0}}{b_{0}}, 1-\frac{c_{0}}{p_{0}+E_i}\right\} -\max\left\{\frac{\hat{p}_{0}}{\hat{b}_{0}}, 1-\frac{\hat{c}_{0}}{\hat{p}_{0}+\hat{E}_i}\right\}\bigg|\nonumber\\
	&\leq \max\left\{\left |\frac{p_{0}}{b_{0}}- \frac{\hat{p}_{0}}{\hat{b}_{0}}\right|,\left | \frac{c_{0}}{p_{0}+E_i}-\frac{\hat{c}_{0}}{\hat{p}_{0}+\hat{E}_i}\right|\right\}\label{eqn:error1} \\
	&\leq\left |\frac{p_{0}}{b_{0}}- \frac{\hat{p}_{0}}{\hat{b}_{0}}\right|+\left| \frac{c_{0}}{p_{0}+E_i}-\frac{\hat{c}_{0}}{\hat{p}_{0}+\hat{E}_i}\right|\\
	&\leq \frac{|p_{0}\hat{b}_{0}-\hat{p}_{0}b_{0}|}{b_{0}\underline{b}}+\frac{|c_{0}\hat{p}_{0}+c_{0}\hat{E}_{i}-\hat{c}_{0}p_{0}-\hat{c}_{0}E_{i}|}{p_{0}\underline{p}} \label{eqn:error2}\\
	&\leq  \frac{\epsilon(p_{0}+b_{0})}{b_{0}\underline{b}}+\frac{\epsilon(c_{0}+p_{0})+c_{0}\bar{E}_{i}+\epsilon E_i}{p_{0}\underline{p}}\label{eqn:error3}\\
	&:= A_i\epsilon,\nonumber
	\end{align}
	where (\ref{eqn:error1}) follows from Lemma~\ref{lemma:max}, (\ref{eqn:error2}) follows from Assumptions~\ref{assump:dist} and $E_{i}>0$, (\ref{eqn:error3}) comes from (\ref{fact}), and $A_{i}:=\frac{p_{0}+b_{0}}{b_{0}b}+\frac{c_{0}+p_{0}+c_{0}e_{i}+b_{0}-p_{0}}{p_{0}\underline{p}}$.
	
	Similarly, we have
	\begin{align}
	&	|l_{i}-\hat{l}_{i}|\nonumber\\
	&\leq  \left|\frac{p_{0}}{b_{0}}- \frac{\hat{p}_{0}}{\hat{b}_{0}}\right|+\left|\frac{c_{0}-E_{i}}{b_{0}-p_{0}-E_{i}}-\frac{\hat{c}_{0}-\hat{E}_{i}}{\hat{b}_{0}-\hat{p}_{0}-\hat{E_{i}}}\right|\label{eqn:err_1}\\
	&\leq \frac{|p_{0}\hat{b}_{0}-\hat{p}_{0}b_{0}|}{b_{0}b}+\nonumber\\
	&\frac{|(c_{0}-E_{i})(\hat{b_{0}}-\hat{p}_{0}-\hat{E}_{i})-(\hat{c}_{0}-\hat{E}_{i})(b_{0}-p_{0}-E_{i})|}{(b_{0}-p_{0}-E_{i}-(2+e_{i})\epsilon)^2}\label{eqn:err_11}\\
	&\leq \epsilon \left(\frac{p_{0}+b_{0}}{b_{0}b}+ (2+e_i)\frac{ |c_{0}-E_{i}|+|b_{0}-p_{0}-E_{i}|}{(b_{0}-p_{0}-E_{i}-(2+e_{i})\epsilon)^2}\right)\label{eqn:err_2}\\
	&:=B_i\epsilon,\nonumber
	\end{align}
	where (\ref{eqn:err_1}) follows from the fact that 
	\begin{align*}
	&|\min\{a,b\}-\min\{c,d\}|=|\!-\!\max\{-a,-b\}\!+\!\max\{-c,-d\}|\\
	&\leq \max\{|-a+c|,|-b+d|\} \leq |a-c|+|b-d|.
	\end{align*} 
	(\ref{eqn:err_11}) then follows from the assumption that $\epsilon$ is sufficiently small such that {$b_{0}-p_{0}-E_{i}-(2+e_{i})\epsilon >\frac{1}{2} b_{0}-p_{0}-E_{i}, \forall i$}. Combining with Lemma~\ref{lemma:E_i}, and the definition of $e_{i}$ in (\ref{eqn:e_i_g}), we have  the bound in (\ref{eqn:err_2}), where $B_i:=\frac{p_{0}+b_{0}}{b_{0}b}+ (2+e_i)\frac{ |c_{0}-E_{i}|+|b_{0}-p_{0}-E_{i}|}{(\frac{1}{2}b_{0}-p_{0}-E_{i})^2} $. %Since we know  $b_{0}-p_{0}> E_{1}$, therefore, for sufficient enough $\epsilon$, our statement holds.
\end{proof}

\begin{Lemma}\label{B(t)}
	% If at time t, we are at exploitation phase. 
	%	 When $L\geq \frac{4K}{\theta_{K}^{2}}$ and $D\geq1$, then with probability at least $1-\frac{1}{t^2}$, 
	Denote $T_p(t)$ as the total number of transmissions during the exploration phase up to frame $t$. Then, the probability that $T_p(t)$ is less than $\frac{\theta_{K}}{2}D(t)$ is at most $\exp\bigg(-\frac{\theta_{K}^2(D(t)-1)}{2K}\bigg)$.
\end{Lemma}
\begin{proof}
	Denote $\Tc(t)$ as the set of time frame indices in the exploration stage up to time frame $t$. Then, we have
	\begin{align}\label{eqn:Tx_bound}
	D(t)\leq |\Tc(t)|<K(D(t)+1),
	\end{align}
	where $D(t):=L\log(t)+D$.
	%First, we note that the total number of time frames that are in the exploration stage, denoted as $n(t)$, is lower bounded by $D(t):=L\log(t)+D$ and upper bounded by $K(D(t)+1)$. 
	%Denote the total number of transmissions in the exploration stage up to time frame $t$ as $T_p(t)$. 
	Then,
	%, is	Using $n(t)$ to denote the number of explorations up to time $t$, we can see $D(t)<n(t)<K(D(t)+1)$,$D(t):=L\log(t)+D$. 
	%	As we can see, for each channel i($\forall i\in [K]$) with probability at least $\theta_{K}>0$, this channel is in good condition, which means that this channel can transmit and get reward. We can know
	\begin{eqnarray}
	&&\Eb[T_p(t)] \geq \theta_{K} |\Tc(t)| \geq \theta_{K}D(t).
	\end{eqnarray}
	Therefore,
	\begin{align}
	&\Pb\left[T_p(t)-\Eb[T_p(t)]< -\frac{\theta_{K}}{2}D(t)\right]\nonumber\\
	&\leq \exp\bigg(-\frac{2(-\frac{\theta_{K}}{2}D(t))^2}{ |\Tc(t)|}\bigg)\label{eqn:hoeff}\\
	&\leq \exp\bigg(-\frac{\frac{\theta_{K}^{2}}{2}D(t)^{2}}{K(D(t)+1)}\bigg)\\
	&\leq  \exp\bigg(-\frac{\theta_{K}^2(D(t)-1)}{2K}\bigg),\label{eqn:tx_bound2}
	%&=& 1-\exp\bigg(-\frac{\theta_{K}^{2}(L\log t)}{2K}\bigg)= 1-t^{-\frac{\theta_{K}^{2}L}{2K}}\\
	%	&\geq & 1-\frac{1}{t^2}, \label{eqn:tx_bound_prob}
	\end{align}
	where $(\ref{eqn:hoeff})$ follows from the Hoeffding's inequality~\cite{Hoeffding:1963}, and (\ref{eqn:tx_bound2}) follows from (\ref{eqn:Tx_bound}). %The inequality in (\ref{eqn:tx_bound_prob}) then follows when $D\geq 1, L\geq \frac{4K}{\theta_{K}^{2}}$.
\end{proof}

\begin{Lemma}\label{summary}
	Denote $\Gamma_1=\frac{1}{2}\min_{i\in [K]}\{\theta_{i}-\theta_{i+1}	\}$,
	\begin{align*}
	\Gamma_2
	&=\min\Bigg\{\left\{\frac{|u_{i}-\theta_{i}|}{A_{i}+1}\right\}_{i: u_i\neq \theta_i} ,\left\{\frac{|l_{i}-\theta_{i}|}{B_{i}+1} \right\}_{i: l_i\neq \theta_i},\\
	&\qquad\frac{b_{0}-p_{0}-E_{1}}{e_{1}},\frac{b_{0}}{2+e_{1}}\Bigg\}\\
	\Gamma&=\min\{\Gamma_{1},\Gamma_{2}\}.
	\end{align*}

	If
	\begin{align}
	L&\geq \max\left\{\frac{1}{\Gamma^{2}},\frac{2\Delta_b^2}{\theta_{K}\Gamma^{2}},\frac{2\Delta_c^2}{K\Gamma^{2}},\frac{2\Delta_p^2}{\theta_{K}\Gamma^{2}}, \frac{4K}{\theta_K^2}\right\},\label{eqn:L}\\
	D&\geq\max\left\{\frac{\log(2K)}{2\Gamma^{2}},\frac{\Delta_b^2\log 2}{\theta_{K}\Gamma^{2}},\frac{\Delta_c^2}{K\Gamma^{2}},\frac{\Delta_p^2\log 2}{\theta_{K}\Gamma^{2}},1\right\},\label{eqn:D}
	\end{align}
	with probability at least $1-5/t^2$, the user makes the same decision as that under the optimal offline policy in frame $t$ if it is in the exploitation phase.
\end{Lemma}	

\begin{proof}
	First, we note that if the following conditions are satisfied, the user must make the 	the same decision as that under the optimal offline policy:
	1) The estimated channel means $\hat{\theta}_i$  are ranked the same as $\theta_i$. 2) For each channel $i$, the order of $\hat{\theta}_i$ and the estimated thresholds $\hat{u}_i$ and $\hat{l}_i$ are the same as that of $\theta_i$, $u_i$, and $l_i$. A sufficient condition to have the first condition hold is to have $|\theta_{i}-\hat{\theta}_{i}|<\Gamma_1,\forall i\in [K]$. 
	For the second condition, we note that if all channels are ranked correctly, and the estimation errors on $\hat{c}_0$, $\hat{b}_0$ and $\hat{p}_0$ are sufficiently small, we can related $|\hat{u}_{i}-u_{i}|$ and $|\hat{l}_{i}-l_{i}|$ with them according to Lemma~\ref{lemma:estimate_error_bound}. 
	
	Specifically, for the general scenario where $\theta_{i}\neq u_{i}, l_{i}$, according to Lemma~\ref{lemma:estimate_error_bound}, if $|\theta_{i}-\hat{\theta}_{i}|<\epsilon $, we must have $\hat{u}_{i} \in (u_{i}-A_{i}\epsilon, u_{i}+A_{i}\epsilon)$.	Thus, when $\epsilon$ is sufficiently small s.t. $(A_{i}+1)\epsilon< |u_{i}-\theta_{i}|$, the order of  $\hat{u}_{i}, \hat{\theta}_{i}$ will be the same as that of $u_{i}, \theta_{i}$. Similarly, we can show that when $\epsilon$ is sufficiently small s.t. $(B_{i}+1)\epsilon< |l_{i}-\theta_{i}|$, the order of $l_{i}, \theta_{i}$ and the order of $\hat{l}_{i}, \hat{\theta}_{i}$ will be the same. 	
	Combining those two cases together, we note that when $|\theta_{i}-\hat{\theta}_{i}|\leq\min\left\{\frac{|u_{i}-\theta_{i}|}{A_{i}+1},\frac{|l_{i}-\theta_{i}|}{B_{i}+1}\right\}$, the transmitter will make the same sensing and transmission decision over channel $i$ as in Theorem~\ref{thm:offline}.
	
	For the scenario where $\theta_{i}=u_{i}$ (or $l_{i}$), according to Theorem~\ref{thm:offline}, this implies that taking actions {\it guess/sense} (or {\it quit/sense}) won't make any difference in the expected reward. Thus, we only need to make sure the order of $\hat{\theta}_i$ and $\hat{l}_i$ (or $\hat{u}_i$) is unchanged in order to make sure the expected reward is the same as that under the optimal policy in Theorem~\ref{thm:offline}.
	
	Combining all conditions together, we note that if all estimation errors are bounded by $\min\{\Gamma_1,\Gamma_2\}$, $\forall i\in[K]$, the user is able to make the same decision as that under the optimal offline policy. 
	Then,
	\begin{align}
	& \Pb[\mbox{user takes suboptimal policy in exploration frame $t$}]\nonumber\\
	&\leq \sum_{i=1}^{K}\Pb[|\theta_{i}-\hat{\theta}_{i}|>\Gamma]+\Pb[|c_{0}-\hat{c}_{0}|>\Gamma]\nonumber\\
	&\qquad+\Pb[|b_{0}-\hat{b}_{0}|>\Gamma]+\Pb[|p_{0}-\hat{p}_{0}|>\Gamma]\nonumber\\	
	&\leq\sum_{i=1}^{K}\Pb[|\theta_{i}-\hat{\theta}_{i}|>\Gamma]+ \Pb[|c_{0}-\hat{c}_{0}|>\Gamma]\nonumber\\
	&\quad+\Pb\left[T_p(t)\leq\frac{\theta_{K}}{2}D(t)\right]\nonumber\\
	&\quad+\Pb\left[|b_{0}-\hat{b}_{0}|\!>\!\Gamma\middle|T_p(t)\!>\!\frac{\theta_{K}}{2}D(t)\right]\Pb\left[T_p(t)\!>\!\frac{\theta_{K}}{2}D(t)\right]\nonumber\\
	&  \quad+\Pb\left[|p_{0}-\hat{p}_{0}|\!>\!\Gamma\middle|T_p(t)\!>\!\frac{\theta_{K}}{2}D(t)\right]\Pb\left[T_p(t)\!>\!\frac{\theta_{K}}{2}D(t)\right]\nonumber\\
	&\leq  2K\exp(-2D(t)\Gamma^2)+2\exp\left(-\frac{2KD(t)\Gamma^2}{\Delta_c^2}\right)\nonumber\\
	&\quad+2\exp\left(-\frac{\theta_{K}D(t)\Gamma^2}{\Delta_b^2}\right)+2\exp\left(-\frac{\theta_{K}D(t)\Gamma^2}{\Delta_p^2}\right)\nonumber\\
	&\quad+\exp\bigg(-\frac{\theta_{K}^2(D(t)-1)}{2K}\bigg)\label{corB(t)}\\
	&\leq \frac{5}{t^2}
	\end{align}
	where in (\ref{corB(t)}) we use the fact that the number of observed probing costs is lower bounded by $KD(t)$ and Lemma \ref{B(t)}, and the last inequality follows when the conditions on $L$ and $D$ in (\ref{eqn:L})(\ref{eqn:D}) are satisfied.
\end{proof}

Now, we are ready to prove Theorem~\ref{thm2}. Denote $\Delta$ as the maximum regret generated in a single frame. Since the maximum net reward in each frame is bounded by $b_{0}-p_{0}$, and the minimum possible reward is above $-p_{0}-K c_{0}$, we have $\Delta \leq b_{0}+K c_{0}$. Therefore,
\begin{align}
& R(T)	\nonumber\\
&\leq \sum_{t \in \Tc(t) } \Delta +\sum_{t\notin \Tc(t) }\Pb[\mbox{take sub-optimal policy in frame $t$}]\Delta  \nonumber\\
&\leq K(D(T)+1)(b_{0}+Kc_{0})+\sum_{t=1}^{T} \left(\frac{6}{t^2} \right) \Delta\label{eqn:regret_bound1}\\
&\leq KL(b_{0}+Kc_{0}) \log T +\left(\pi^2+DK+1\right)(b_{0}+Kc_{0})\nonumber\\
&:= C_1\log T +C_2\nonumber
\end{align}
where $C_1:=KL(b_{0}+Kc_{0})$, $C_2:=\left(\pi^2+DK+1\right)(b_{0}+Kc_{0})$.

\subsection{Proof of Theorem~\ref{thm:lower_bound}}\label{appx:thm3}
Denote $\Xv(t)$ as the status of channels observed in frame $t$. It contains the status of channels sensed or guessed in frame $t$. Let $f(\Xv(t);\thetav)$ be the probability mass function (PMF) of observation $\Xv(t)$ under $\thetav$ in frame $t$. 
%And $X_{k}(t)$ denotes the status of channel $k$ in frame $t$ and it obeys bernoulli distribution with expectation $\theta_k$.

Consider another parameter setting $\thetav' = \{\theta_1, \theta_2, \ldots, \theta_k', \ldots, \theta_K\}$. The only difference between $\thetav$ and $\thetav'$ is $\theta_k\neq\theta_k'$, where $k \in \{K^*+1, \ldots, K\}$.
We assume $1>\theta_k'>\theta_1$. Thus, with the new parameter vector $\thetav'$, the user would first examine channel $k'$. Let $u_k'$ be the corresponding upper threshold that $\theta_k'$ would be compared with. We assume $\theta_k'>u_k'$, i.e., the optimal offline policy with $\thetav'$ is to transmit over channel $k$ without sensing (i.e., {\it guess}). Besides, we assume
%\begin{align}
$(1+\epsilon)kl(\theta_k,\max\{\theta_1, u_k'\}) = kl(\theta_k,\theta'_k)$ 
%\end{align}
where $\epsilon>0$. 
Denote $\lv_{x}$ as an indicator function, which equals one if $x$ is true.
Then, we have
\begin{align}
&\frac{f(\Xv(t); \thetav)}{f(\Xv(t);\thetav')} = \lv_{X_k(t) \in \Xv(t)} \frac{X_k(t)\theta_k + (1 - X_k(t))(1 - \theta_k)}{X_k(t) \theta_k' + (1-X_k(t))(1-\theta_k')}\nonumber\\
&\qquad\qquad\quad+ \lv_{X_k(t) \notin \Xv(t)} 
\end{align}

Let $\Eb_{\thetav}$, $\Pb_{\thetav}$ be the expectation and probability measure associated with $\thetav$. Then,
\begin{eqnarray}
d\Pb_{\thetav} &=& \prod_{t=1}^T\frac{f(\Xv(t);\thetav)}{f(\Xv(t);\thetav')}d\Pb_{\thetav'} \nonumber\\
&=& \prod_{t=1}^T \lv_{X_k (t)\in \Xv(t)} \frac{X_k(t)\theta_k + (1 - X_k(t))(1 - \theta_k)}{X_k(t) \theta_k' + (1-X_k(t))(1-\theta_k')}d\Pb_{\thetav'} \nonumber\\
&=& \exp \left(\hat{kl}_{\hat{N}_{k}(T)}\right) d\Pb_{\thetav'} \label{eqn:lower1}
\end{eqnarray}
where $\hat{N}_{k}(T)$ denotes the number of frames that channel $k$ has been observed up to frame $T$, and $\hat{kl}_{s} := \sum_{t=1}^s \log\frac{X_k(t)\theta_k + (1 - X_k(t))(1 - \theta_k)}{X_k(t) \theta_k' + (1-X_k(t))(1-\theta_k')}.$

Let $N_{k}(T)$ be the number of times channel $k$ is involved in the policy up to time frame $T$. Note that $N_{k}(T)$ is in general greater than $\hat{N}_{k}(T)$, as for some cases, the policy may meet some stopping conditions before channel $k$ is observed. Let $f_T =  \frac{(1-\epsilon)\log T}{kl(\theta_k,\theta'_k)}$.
Denote 
%\begin{equation}
$  \Cc_T := \left\{ N_k(T) < f_T \right\} \cap \left\{ \hat{kl}_{\hat{N}_k(T)} \leq \left(1 - \frac{\epsilon}{2}\right) \log T \right\} $. %\label{eqn:lower2}
%\end{equation}
Based on (\ref{eqn:lower1}), we have
\begin{eqnarray}
\Pb_{\thetav}[\Cc_T] &=& \int_{\Cc_T} \exp \left(\hat{kl}_{\hat{N}_{k}(T)}\right) d\Pb_{\thetav'} \label{eqn:lower3}\\
&\leq& \exp\left(\left(1-\frac{\epsilon}{2}\right)\log T \right)\Pb_{\thetav'}[\Cc_T] \label{eqn:lower4}\\
&\leq& T^{1-\frac{\epsilon}{2}} \Pb_{\thetav'}[N_k(T) < f_T] \label{eqn:lower5}\\
&\leq& T^{1-\frac{\epsilon}{2}} \frac{\Eb_{\thetav'}[T - N_k(T)]}{T - f_T} \label{eqn:lower6}\\
&=& \frac{\Eb_{\thetav'}[T - N_k(T)]}{T^{\frac{\epsilon}{2}}\left(1 - \frac{f_T}{T}\right)} = o(1) \label{eqn:lower7}
\end{eqnarray}
where (\ref{eqn:lower5}) follows from the fact that $\Pb_{\thetav'}[\Cc_T] \leq \Pb_{\thetav'}[N_k(T) < f_T]$, (\ref{eqn:lower6}) is based on Markov's inequality, and (\ref{eqn:lower7}) follows from the fact that $N_k(T)$ is always greater than or equal to the number of times that {\it guess} on channel $k$ is chosen. Since {\it guess} on channel $k$ is the optimal offline policy with $\thetav'$, following the definition of $\alpha$-consistent strategy in Definition~\ref{dfn:alpha_consistent}, we must have (\ref{eqn:lower7}) hold.

Furthermore, by using the fact that $\hat{N}_{k}(T) \leq N_{k}(T)$, we have
\begin{eqnarray}
\Pb_{\thetav}[\Cc_T] \geq \Pb_{\thetav}\left[N_k(T) < f_T , \max_{t \leq f_T} \hat{kl}_{t} \leq \left(1 - \frac{\epsilon}{2}\right)\log T \right] \nonumber\\
= \Pb_{\thetav}\left[N_k(T) < f_T, \frac{1}{f_T} \max_{t \leq f_T}\hat{kl}_{t} \leq \frac{\left(1 - \frac{\epsilon}{2}\right)kl(\theta_k;\theta'_k)}{(1 -\epsilon)} \right] \label{eqn:lower8} 
\end{eqnarray}
According to the maximal law of large numbers~\cite{Lai:1985:AEA}, we have
\begin{eqnarray}
\lim_{T \rightarrow \infty} \frac{1}{T} \max_{t=1,2,\ldots,T} \hat{kl}_{t} &=& kl(\theta_k;\theta_k') \qquad\mbox{a.s.}\label{eqn:large_number}\nonumber
\end{eqnarray}
Plugging it in (\ref{eqn:lower8}), we get
\begin{equation}
\lim_{T \rightarrow \infty}\Pb_{\thetav}[\Cc_T] \geq \lim_{T \rightarrow \infty}\Pb_{\thetav}[N_k(T) < f_{T}] \label{eqn:lower9}
\end{equation}
Combining (\ref{eqn:lower9}) with (\ref{eqn:lower7}), we have
%\begin{align}
$\lim_{T \rightarrow \infty}\Pb_{\thetav}[N_k(T) \geq f_{T}]=1.$
%\end{align}
Since $R(T) \geq \min_{\pi \in \Pi_k}\Delta_\pi \Eb_{\thetav}[N_k(T)]$ for all $k \in [K^*+1, \ldots, K]$, %  \label{eqn:lower10}
we have
\begin{eqnarray}
\liminf_{T\rightarrow \infty} \frac{R(T)}{\log T} &\geq& \max_{k=K^*+1,\ldots,K}\left\{\min_{\pi \in \Pi_k} \frac{\Delta_\pi}{kl(\theta_k;\max(\theta_1,u'_k))}\right\}.\nonumber
\end{eqnarray}

\bibliographystyle{IEEEtran}
\bibliography{CognRadio}

\end{document}